\documentclass[sigconf,natbib=true]{acmart}
\usepackage{amsthm}
\usepackage{multirow}
\usepackage{graphicx}
\usepackage{subcaption}
\theoremstyle{remark}
\newtheorem{remark}{Remark}
\AtBeginDocument{%
  }

\setcopyright{acmlicensed}
\copyrightyear{2018}
\acmYear{2018}
\acmDOI{XXXXXXX.XXXXXXX}
\acmConference[Conference acronym 'XX]{Make sure to enter the correct
  conference title from your rights confirmation email}{June 03--05,
  2018}{Woodstock, NY}
\acmISBN{978-1-4503-XXXX-X/2018/06}

\author{Junchen Fu} 
\affiliation{
  \institution{University of Glasgow}\streetaddress{}\city{Glasgow}\country{United Kingdom}}
\email{j.fu.3@research.gla.ac.uk}
  
\author{Xuri Ge}
\affiliation{
\institution{Shandong University}\streetaddress{}\city{Jinan}\country{China}}
\email{xuri.ge@sdu.edu.cn}

\author{Alexandros Karatzoglou}
\affiliation{
\institution{Amazon}\streetaddress{}\city{Barcelona}\country{Spain}}
\email{alexandros.karatzoglou@gmail.com}

\author{Ioannis Arapakis}
\affiliation{
\institution{Telef\'{o}nica Scientific Research }\streetaddress{}\city{Barcelona}\country{Spain}}
\email{arapakis.ioannis@gmail.com}

\author{Suzan Verberne}\affiliation{
\institution{Leiden University}\streetaddress{}\city{Leiden}\country{Netherlands}}
\email{s.verberne@liacs.leidenuniv.nl}

\author{Joemon M. Jose}\affiliation{
\institution{University of Glasgow}\streetaddress{}\city{Glasgow}\country{United Kingdom}}
\email{joemon.jose@glasgow.ac.uk}

\author{Zhaochun Ren}
\affiliation{
  \institution{Leiden University}
  \city{Leiden}
  \country{Netherlands}
}
\email{z.ren@liacs.leidenuniv.nl}

\makeatletter
\g@addto@macro\addresses{\@@authornotemark{2}} 
\g@addto@macro\@authornotes{\footnotetext[2]{Corresponding author.}}
\makeatother
\begin{document}

\title[Differentiable Semantic ID for Generative Recommendation]{Differentiable Semantic ID for Generative Recommendation}

\begin{abstract}
Generative recommendation provides a novel paradigm in which each item is represented by a discrete semantic ID (SID) learned from rich content. 
Most methods treat SIDs as predefined and train recommenders under static indexing. In practice, SIDs are optimized only for content reconstruction rather than recommendation accuracy. This leads to an \textit{objective mismatch}: the system optimizes an indexing loss to learn the SID, and a recommendation loss for interaction prediction, but because the tokenizer is trained independently, the recommendation loss cannot update it.
A natural approach is to make semantic indexing differentiable so recommendation gradients can directly influence SID learning, but this often causes codebook collapse with only a few codes used. We attribute this to early deterministic assignments that limit codebook exploration, leading to imbalance and unstable optimization.

In this paper, we therefore propose \textbf{DIGER} (\underline{D}ifferentiable Semantic \underline{I}D for \underline{GE}nerative \underline{R}ecommendation). 
DIGER is a first step towards an effective differentiable semantic ID for generative recommendation. The Gumbel noise explicitly encourages early-stage exploration over codes, mitigating collapse and improving code utilization. To better balance exploration and convergence, we introduce two uncertainty decay strategies that reduce the Gumbel noise, enabling a gradual shift from early-stage exploration to the exploitation of learned SIDs. 
Extensive experiments across multiple public datasets demonstrate consistent improvements from differentiable semantic ID. 
These results confirm the effectiveness of aligning indexing and recommendation objectives through differentiable SIDs. This identifies differentiable SID as a promising area of study.
Our code is released under~\url{https://github.com/junchen-fu/DIGER}.
\end{abstract}

\keywords{Generative Recommendation; Differentiable Semantic ID; Exploration and Exploitation; DRIL; Uncertainty Decay}

\maketitle

\section{Introduction}
\begin{figure}[t]
  \centering
  \begin{subfigure}{\linewidth}
    \centering
    \includegraphics[width=\linewidth]{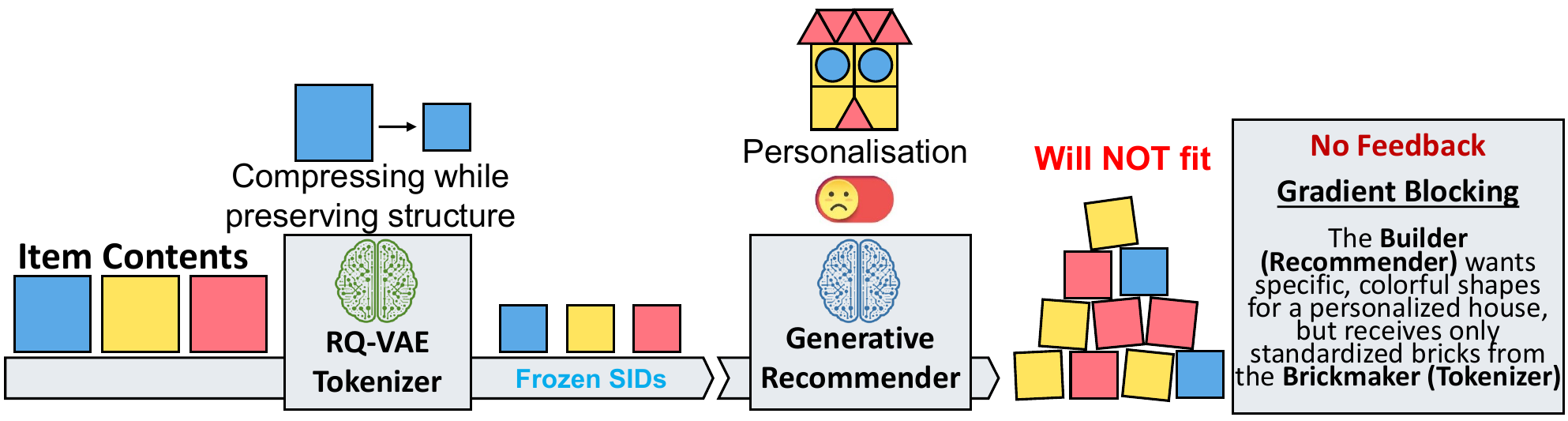}
    \subcaption{Conventional pipeline freezes RQ-VAE semantic IDs, causing gradient blocking from the recommender and objective mismatch.}
    \label{fig:motivation:a}
  \end{subfigure}

  \begin{subfigure}{\linewidth}
    \centering
    \includegraphics[width=\linewidth]{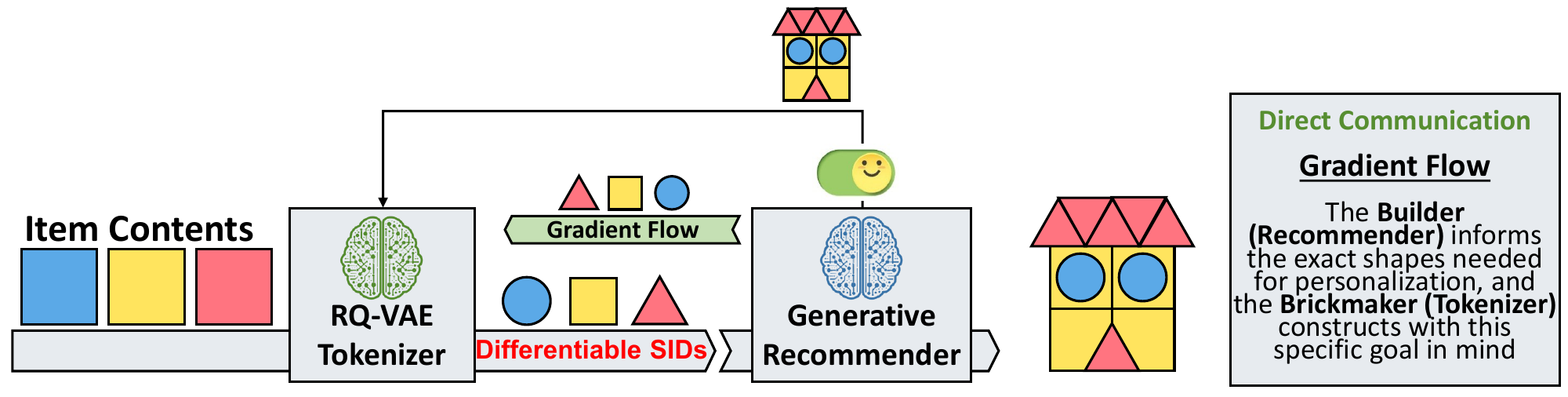}
    \subcaption{Our differentiable semantic IDs enable gradient flow for joint optimization toward recommendation utility.}
    \vspace{-0.15in}
    \label{fig:motivation:b}
  \end{subfigure}

  \caption{Conventional vs.\ Generative recommendation with Differentiable SID. A “brickmaker--builder” sketch is used to illustrate the tokenizer and recommender.}
  \vspace{-0.2in}
  \label{fig:motivation}
\end{figure}

Generative recommendation has introduced a paradigm shift in modern recommender systems~\cite{rajput2023recommender,geng2022recommendation,zhao2026unifying}. Rather than representing items with continuous embeddings alone, recent methods leverage rich item content, most notably textual descriptions, to construct discrete semantic identifiers (SIDs)~\cite{tay2022transformer,wang2025generative,han2025mtgr,hou2025generative} and thereafter formulate recommendation as next-SID generation. This design introduces desirable properties, such as compact indexing, semantic interpretability, and a natural interface to sequence modeling~\cite{hou2025survey}.

Despite these advantages, current practices still suffer from a mismatch problem. First, a vector-quantized model such as RQ-VAE~\cite{lee2022autoregressive} is trained with a reconstruction objective to learn a discrete codebook and assign each item a fixed sequence of semantic IDs. Second, given a user's interaction history, a sequence generator (e.g., a Transformer) is trained to predict the next-item semantic IDs \cite{rajput2023recommender,tay2022transformer}. 
While this pipeline is effective, the indexing loss and the recommendation loss are optimized independently. It introduces a fundamental issue: \emph{semantic indexing is not aligned with the downstream recommendation objective}. This mismatch can degrade recommendation performance, since reconstruction-oriented semantic IDs are not optimized for ranking and remain predefined during user interest evolution~\cite{ye2025harnessing}, limiting personalization and preference-aware representation learning.

Since semantic IDs are precomputed and frozen, the recommendation loss cannot propagate gradients back to the codebook or the ID assignment. In other words, the indexing space is optimized for content reconstruction rather than for next-item prediction. As illustrated in Figure~\ref{fig:motivation}(a), the tokenizer learns semantic IDs for reconstruction, not personalization. Because these SIDs are frozen, recommendation gradients cannot reach the tokenizer, so the recommender is stuck with reconstruction-oriented codes that are misaligned with recommendation utility. 
Consequently, the overall effectiveness of generative recommendations is constrained due to this misalignment.

An intuitive solution to the objective mismatch is to make semantic indexing differentiable, allowing recommendation gradients to jointly optimize the encoder and the generative semantic IDs (Figure~\ref{fig:motivation}(b)), similar to MoRec~\cite{yuan2023go,li2023exploring} in traditional modality-based recommendation, where jointly training BERT-based item encoders with recommendation objectives significantly outperforms frozen features.
Yet, directly differentiable semantic IDs for semantic indexing remain underexplored~\cite{hou2025survey}. 
A main challenge in passing gradients through a codebook is its discrete nature: backpropagating through discrete SIDs is inherently non-trivial. The most common workaround --- the straight-through estimator (STE)~\cite{bengio2013estimating} used in RQ-VAE training~\cite{lee2022autoregressive} --- is unstable in practice and can trigger semantic ID collapse~\cite{huh2023straightening}. As shown in Figure~\ref{fig:gumbel-vs-ste}, naive differentiable optimization with STE quickly becomes over-confident in the early stages of training, causing a small subset of codes to dominate, while the rest are rarely selected, i.e., severely reduced code utilization. This collapse is further observed along with large training-time variance: STE exhibits a markedly larger standard deviation in code balance (error bars) across codebook levels, indicating volatile and inconsistent code usage throughout training. Such an imbalance directly harms recommendation quality (e.g., NDCG@10), making ``naively'' differentiable semantic indexing brittle and frequently under-performing.

\begin{figure}[t]
    \centering
    \includegraphics[width=\linewidth]{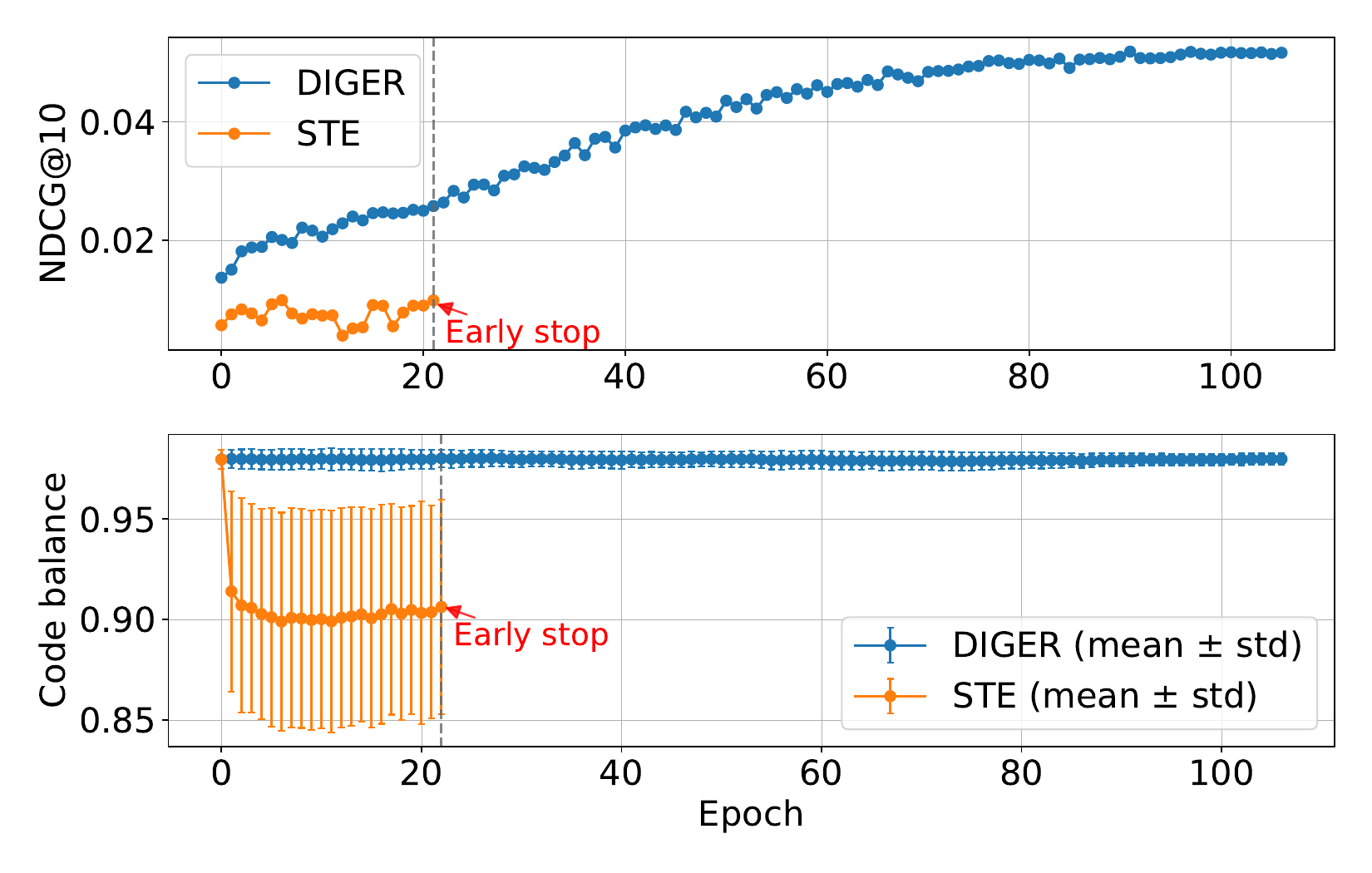}
    \vspace{-0.25in}
    \caption{Comparison of DIGER and STE on B-Shop Dataset. Top: validation NDCG@10 across epochs. Bottom: code balance summarized by the mean over three codebook levels' coverage, with error bars indicating the standard deviation across levels. Grey dashed vertical line marks the STE early-stop epoch.}
    \label{fig:gumbel-vs-ste}
    \vspace{-0.2in}
\end{figure}

In this work, we propose \textbf{DIGER} (\underline{D}ifferentiable Semantic \underline{I}D for \underline{GE}nerative \underline{R}ecommendation), which enables semantic IDs to be learned jointly with a generative recommender\footnote{DIGER is metaphorically used as a digger that breaks through and opens a direct gradient pathway between the generative recommender and the semantic ID.}. To enable stable joint training, we design a differentiable SID framework with exploratory learning (DRIL) with Gumbel noise~\cite{jang2016categorical} to address the codebook collapse issues. Motivated by the exploration-exploitation paradigm~\cite{gupta2006interplay}, we further propose two uncertainty decay strategies which gradually reduce the stochasticity induced by the injected noise to better align with training objective.
As shown in Figure~\ref{fig:gumbel-vs-ste}, training remains stable, with balanced code usage and steadily improving performance. The injected Gumbel noise promotes early-stage exploration over codes, alleviating collapse and improving codebook utilization. This design enables stable joint optimization and consistently improves next-SID generation. Our contributions are as follows:
\begin{itemize}
\item We propose DIGER, an effective differentiable semantic indexing framework. To the best of our knowledge, it represents a pioneering effort to effectively enable direct joint optimization of semantic IDs and generative recommenders.
\item To realize DIGER, we introduce a two-part paradigm consisting of Differentiable Semantic ID with Exploratory Learning (DRIL), an exploration–exploitation–inspired framework equipped with two Uncertainty Decay (UD) strategies.
\item Experiments on three public datasets demonstrate the effectiveness of jointly training with differentiable semantic IDs.
\end{itemize}

\vspace{-0.1in}
\section{Related Work}

\noindent \textbf{Joint Optimization for Recommendation.}
 Jointly training both item and user encoders for recommendation has been extensively studied in modality-based sequential recommender systems, motivated by the objective mismatch between pretrained representation encoders and sequential recommendation tasks~\cite{yuan2023go,wu2021empowering,elsayed2022end,fu2024exploring,fu2024iisan}. A central theme is to adapt modality encoders jointly with the recommender so that learned representations are optimized for sequential prediction. Representative efforts include joint training of visual backbones, such as ResNet~\cite{elsayed2022end}, and fine-tuning language encoders like BERT, to better align textual representations with user--item interactions~\cite{wu2021empowering}. Recent work has extended this paradigm by leveraging large language models (LLMs) as task backbones or feature generators for recommendation~\cite{li2023exploring,liu2024once,fu2025efficient}. Altogether, these studies leverage joint optimization to address the representation--task misalignment challenge, by explicitly allowing recommendation gradients to flow into the representation module (e.g., connecting item encoders with user modeling for joint backpropagation), thereby aligning learned features with the downstream objective, while also highlighting practical optimization challenges when the encoder is large or contains discrete components.

\noindent \textbf{Generative Recommendation.}
Generative recommender systems formulate recommendation as a sequence generation problem over discrete item identifiers~\cite{hou2025generating,wang2025generative,han2025mtgr,hou2025generative,li2025survey,liu2025onerec,wang2025empowering,lin2025order,zhao2025unifying,zhang2025killing,li2025dimerec,zhai2025simple,wang2025generative2,liu2025discrec,zhang2025c2t,zheng2025pre}. To improve generalization, recent approaches replace random IDs with content-informed, semantic IDs. For example, P5~\cite{geng2022recommendation} unifies heterogeneous recommendation signals into natural-language sequences and adopts a text-to-text model for recommendation. TIGER~\cite{rajput2023recommender} introduces a Transformer-based indexing mechanism, where an RQ-VAE codebook assigns each item a sequence of semantic tokens that are subsequently predicted by a decoder. LETTER~\cite{wang2024learnable} further improves semantic ID quality by injecting hierarchical semantics and collaborative signals into residual vector quantization. However, relatively few studies have investigated the objective mismatch between the semantic ID tokenizer and the downstream generative recommender.

Recently, ETEGRec \cite{liu2024generative} proposed a distillation-based approach for aligning tokenizers with generative recommendation models. However, rather than enabling joint optimization via differentiable semantic IDs, it relies on an alternating distillation scheme to handle the discreteness of semantic IDs, where the recommender and the semantic ID tokenizer are trained by distilling knowledge into each other in separate stages. Since the recommendation loss cannot be directly backpropagated through the discrete indexing process, the indexing module is optimized only through indirect surrogate objectives. 
Consequently, the core mismatch caused by the blocked gradient flow remains unresolved. 

\smallskip
\noindent In contrast, our work enables \emph{direct} joint learning for generative recommendation by making semantic IDs differentiable, so that the semantic indexing mechanism and the recommender can be optimized jointly throughout training.

\begin{figure*}[htbp]
    \centering
    \includegraphics[width=0.9\textwidth]{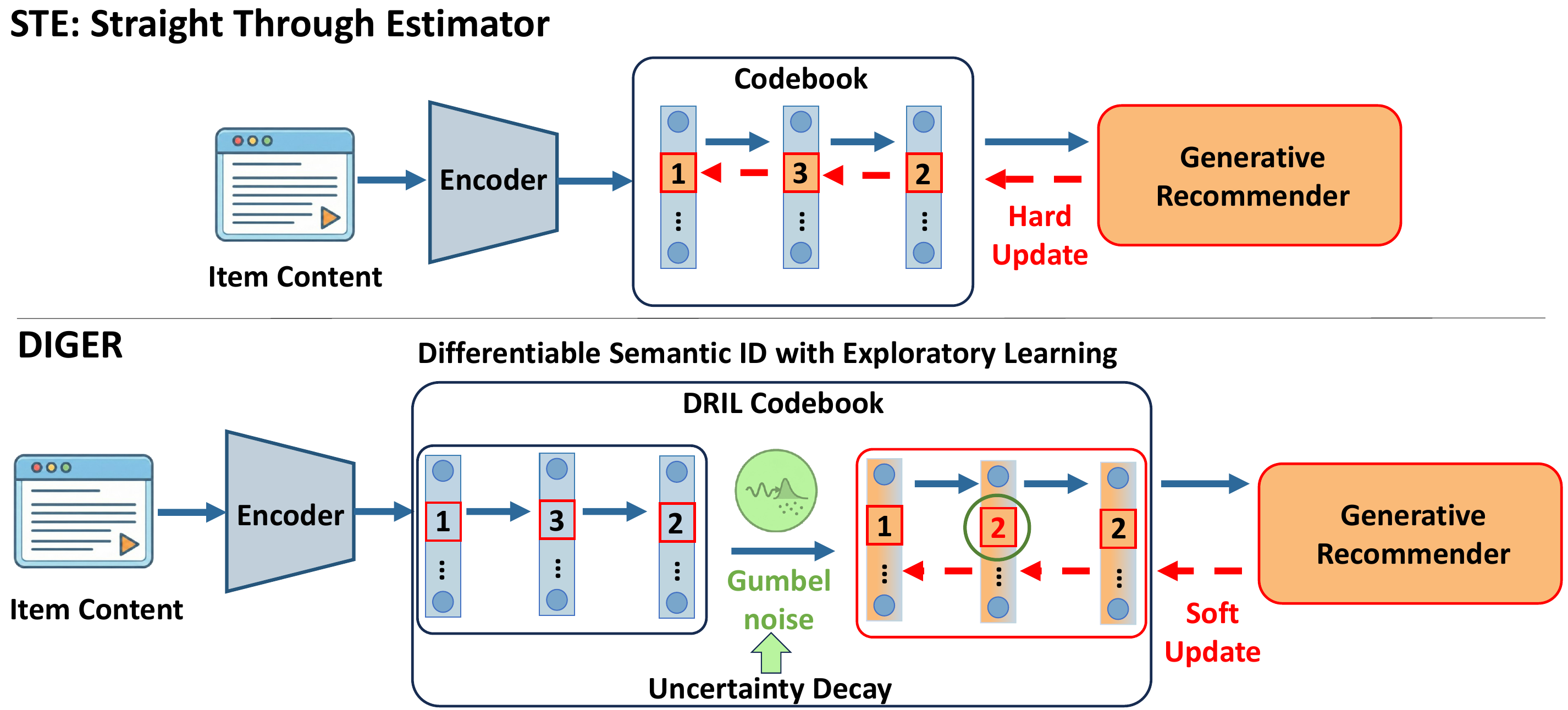}
    \vspace{-0.1in}
    \caption{
Comparison of STE vs. DIGER. 
\textbf{STE:} Relies on deterministic hard selection (e.g., sequence 1-3-2), where gradients only backpropagate to the selected indices (\textit{Hard Update}). 
\textbf{DIGER:} Introduces \textit{Gumbel noise} to encourage exploratory learning. For instance, the noise may shift a deterministic assignment from code 3 to 2 (changing the path from 1-3-2 to 1-2-2), exploring alternative semantic representations. Crucially, DIGER employs \textit{Soft Update}, allowing gradients to flow to \textbf{all} codebook weighted by their Gumbel-Softmax probabilities, while the noise level is progressively reduced via uncertainty decay strategies.
}
\vspace{-0.15in}
    \label{fig:architecture_overview}
\end{figure*}

\begin{figure}[t]
    \centering
    \includegraphics[width=0.9\linewidth]{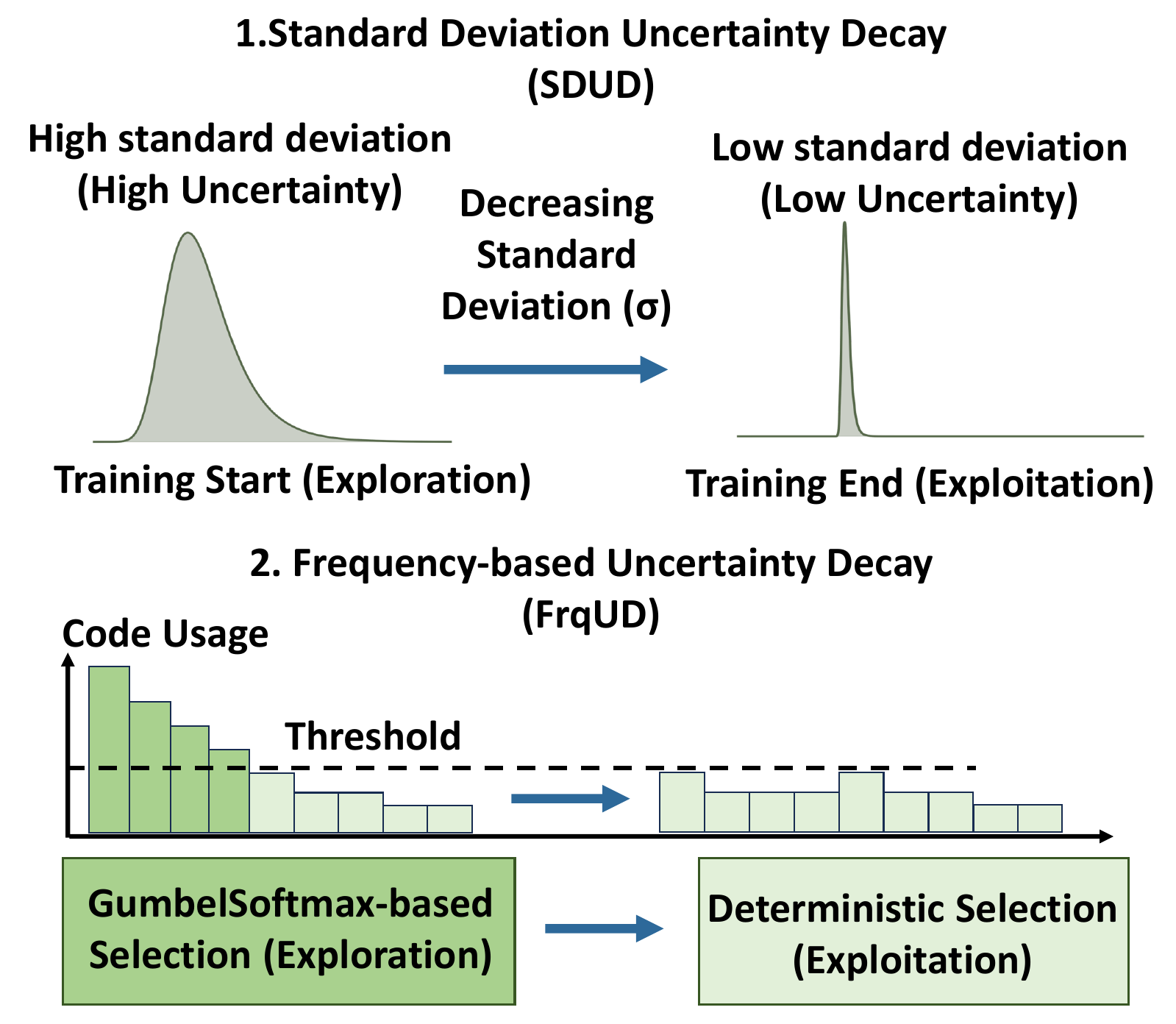}
    \vspace{-0.1in}
    \caption{Two uncertainty decay strategies for exploration--exploitation: standard deviation--based uncertainty decay (SDUD) and frequency-based uncertainty decay (FrqUD).}
    \label{fig:UD}
    \vspace{-0.2in}
\end{figure}

\vspace{-0.08in}
\section{Preliminaries}
\label{sec:preliminary}

Following \cite{rajput2023recommender}, we formulate generative recommendation as an auto-regressive \emph{next semantic ID (SID) prediction problem}, i.e., \emph{given a user's interaction history, a generative model predicts the SID sequence of the next item the user will interact with}.

Let $\mathcal{V}$ denote the set of items. 
We obtain discrete SIDs for items using an RQ-VAE based tokenizer: given an item $v\in\mathcal{V}$ (e.g., its content features), the tokenizer maps it to a length-$m$ discrete code sequence (its Semantic ID):
\begin{equation}
\mathbf{z}_v = (c_{v,1},\ldots,c_{v,m}), \quad c_{v,j}\in\{1,\ldots,K\},
\end{equation}
where $K$ is the SID vocabulary size. Each code $c_{v,j}$ corresponds to an entry in a shared embedding/codebook and is used as the discrete token fed into the generative model.

For a user $u$, suppose the interaction history up to time $t$ is $\mathbf{h}_u^{(t)}=(v_1,\ldots,v_t)$. We convert the history into a token sequence by concatenating the SIDs of previously interacted items in chronological order:
\begin{equation}
\mathbf{x}_u^{(t)} = \mathbf{z}_{v_1}\oplus\cdots\oplus \mathbf{z}_{v_t},
\end{equation}
where $\oplus$ denotes sequence concatenation. The generative recommender is a sequence prediction model (e.g., a Transformer decoder) that learns the conditional distribution of the next-item SID: $p_\theta(\mathbf{z}_{v_{t+1}} \mid \mathbf{x}_u^{(t)})$,
and produces the SID in an auto-regressive manner:
\begin{equation}
p_\theta(\mathbf{z}_{v_{t+1}} \mid \mathbf{x}_u^{(t)})=\prod_{j=1}^{m} 
p_\theta\!\left(c_{v_{t+1},j}\mid \mathbf{x}_u^{(t)}, c_{v_{t+1},<j}\right).
\end{equation}
where the index $j \in \{1,\ldots,m\}$ denotes the position of a discrete code within the Semantic ID sequence of the predicted item $v_{t+1}$.

Training minimizes the negative log-likelihood of the ground-truth next-item SID:
\begin{equation}
\mathcal{L}_{\text{gen}} = -\sum_{u}\sum_{t} \log p_\theta(\mathbf{z}_{v_{t+1}} \mid \mathbf{x}_u^{(t)}).
\end{equation}
At inference time, the model generates the next SID, which realizes the \textit{Next SID Prediction} process, and serves as the basis for next-item recommendation. In particular, the generated SID can be deterministically mapped back to an item via the codebook (or a retrieval step over items sharing the same SID), enabling efficient next-item recommendation.

\vspace{-0.08in}
\section{Our Method: DIGER}
Conventional generative recommendation~\cite{rajput2023recommender} typically adopts a two-stage training paradigm. In contrast, we show that joint optimization operates over a strictly larger feasible solution space, while objective mismatch in the two-stage setting can induce arbitrarily large suboptimality even under exact optimization (Appendix~\ref{sec:appendix}, Remark~\ref{remark:jo}).
 To effectively achieve joint optimization via direct gradient flow, we propose \textbf{DIGER} (\underline{D}ifferentiable Semantic \underline{I}D for \underline{GE}nerative \underline{R}ecommendation). 
DIGER is trained via a two-part learning scheme inspired by the exploration--exploitation paradigm: 
\textbf{DRIL} (\underline{D}iffe\underline{R}entiable Semantic \underline{I}D with Exploratory \underline{L}earning) provides Gumbel-noise--based stochastic exploration to learn discrete semantic IDs with direct gradient flow, 
and \textbf{Uncertainty Decay} progressively reduces exploration and steers the model toward stable exploitation.
 Figure~\ref{fig:architecture_overview} illustrates the overall framework.

\vspace{-0.08in}

\subsection{DRIL: Differentiable Semantic ID with Exploratory Learning}

Directly applying the straight-through estimator (STE) yields deterministic hard assignments, which often lead to overconfident early decisions and premature codebook collapse (Figure \ref{fig:gumbel-vs-ste}), resulting in imbalanced code usage and unstable optimization. From a distributional viewpoint, such deterministic assignments correspond to highly concentrated, low-entropy selection distributions that limit exploration over the codebook. In Appendix \ref{sec:appendix}, Remark \ref{remark:entropy}, we show that increasing assignment entropy increases the effective number of utilized codes, providing a principled motivation for entropy-aware design. Based on this insight, we inject noise into the assignment logits to spread probability mass across multiple codes, and propose a noise-driven assignment mechanism, DRIL, to encourage controlled exploration.

Compared to common Gaussian noise~\cite{lu2020mean}, we adopt Gumbel noise~\cite{jang2016categorical} better captures the probabilistic nature of categorical sampling, inducing a softmax-based selection process in which higher-similarity codes are exponentially more likely to be selected. This aligns more closely with the actual selection process over a codebook, where stronger matches naturally receive higher sampling probabilities. By injecting Gumbel noise into the logits, our method enables stochasticity in code selection, facilitating early-stage exploration while preserving differentiability for end-to-end optimization.

For an item $v$, let the (LLM $\rightarrow$ RQ-VAE encoder) pipeline output a continuous representation $\mathbf{r}_v$ (or $\mathbf{r}_{v,j}$ for the $j$-th code position under position-wise quantization). For each position $j\in\{1,\ldots,m\}$, we compute similarity logits to the codebook $\{\mathbf{e}_i\}_{i=1}^{K}$:
\begin{equation}
\ell_{v,j,i} = \mathrm{sim}(\mathbf{r}_{v,j}, \mathbf{e}_i), \quad i\in\{1,\ldots,K\}.
\end{equation}
To add the Gumbel noise, we then obtain a Gumbel-Softmax distribution:
\begin{equation}
\tilde{y}_{v,j,i}=
\frac{\exp\!\left((\ell_{v,j,i}+ g_{v,j,i})/\tau\right)}
{\sum_{k=1}^{K}\exp\!\left((\ell_{v,j,k}+ g_{v,j,k})/\tau\right)},
\end{equation}
where $g_{v,j,i}\sim \mathrm{Gumbel}(0,1)$ are i.i.d., $\tau$ is the temperature.

\noindent \textbf{Hard SID for forward indexing.}
In the forward pass, we take a hard code by argmax over the noisy logits (equivalently over $\tilde{\mathbf{y}}_{v,j}$):
\begin{equation}
c_{v,j} = \arg\max_{i}\,(\ell_{v,j,i}+ g_{v,j,i}),
\end{equation}
yielding the Semantic ID $\mathbf{z}_v=(c_{v,1},\ldots,c_{v,m})$, which is used to index the embedding table in the generative recommender (Section~\ref{sec:preliminary}).

\noindent \textbf{Soft update for backpropagation.} 
Gradients are computed through the soft probabilities $\tilde{\mathbf{y}}_{v,j}$. Concretely, we use the soft embedding
\begin{equation}
\bar{\mathbf{e}}_{v,j}=\sum_{i=1}^{K}\tilde{y}_{v,j,i}\,\mathbf{e}_i
\end{equation}
to perform a \emph{soft} update of the RQ-VAE codebook/embedding table, while keeping the hard $c_{v,j}$ to form the discrete SID token sequence. The Gumbel noise encourages exploration and improves codebook utilization. 

\vspace{-0.1in}
\subsection{Two Uncertainty Decay Strategies}
 DRIL promotes exploration over discrete codes by sampling diverse code selections during training, but the injected noise perturbs the forward pass and can create a mismatch with inference.  Since validation and inference use deterministic semantic ID assignment (for example, a hard argmax over code similarities) for stability and reproducibility, large Gumbel noise in training later can cause the semantic IDs effectively used during training to differ from those used at inference. Therefore, we adopt an exploration-to-exploitation schedule that reduces the uncertainty introduced by Gumbel noise in the later stage of training, better aligning the training objective with the inference-time objective. To this end, we propose two uncertainty decay strategies to reduce the uncertainty\footnote{In many Gumbel-Softmax applications (e.g., computer vision and speech)~\cite{jang2016categorical,takida2022sq}, the discrete token is not the end goal; it is a means to quantize/compress representations. Accordingly, prior work typically anneals the temperature $\tau$ so that assignments become increasingly one-hot. In our setting, however, the semantic ID itself is the target and must be produced as a discrete output at inference, so temperature annealing is less effective. We instead keep $\tau$ fixed and reduce the uncertainty toward zero. Empirically, varying $\tau$ has little effect on our results.}:

\noindent \textbf{(1) Standard Deviation Uncertainty Decay (SDUD).}
Gumbel noise's standard deviation $\sigma$ controls the stochasticity of SID assignment.
A larger $\sigma$ corresponds to higher uncertainty and encourages exploration,
whereas $\sigma \!\approx\! 0$ yields near-deterministic assignments that are better
aligned with inference-time selection.
Therefore, we design an auxiliary objective whose \emph{stationary point}~\cite{graves2011practical} provides an explicit link between the task loss and the optimal noise level, making the training dynamics more interpretable.

Specifically, we couple the main generative objective $\mathcal{L}_{\text{gen}}$ with $\sigma$ as:
\begin{equation}
\mathcal{L}_{\sigma}
=
\frac{\mathcal{L}_{\text{gen}}}{2(\sigma+\lambda)^2}
+
\log(\sigma+\lambda),
\end{equation}
where $\mathcal{L}_{\text{gen}}$ is the next-SID prediction loss,
$\sigma \ge 0$ is the noise scale (learned through optimization),
and $\lambda>0$ is a hyperparameter. Importantly, $\lambda$ not only avoids numerical issues when $\sigma\!\to\!0$, but also determines when $\sigma$ can effectively reach $0$ (i.e., the desired turning point from exploration to exploitation). In practice, we search $\lambda$ so that $\sigma\!=\!0$ happens near the optimal transition stage.

Taking the partial derivative w.r.t.\ $\sigma$ and setting it to zero gives a closed-form equilibrium (we detail the steps in \autoref{sec:derive_of_sigma}):
\[
\frac{\partial \mathcal{L}_{\sigma}}{\partial \sigma}=0
\quad\Rightarrow\quad
(\sigma+\lambda)^2=\mathcal{L}_{\text{gen}}
\quad\Rightarrow\quad
\sigma^\star = \max\!\big(0,\sqrt{\mathcal{L}_{\text{gen}}}-\lambda\big).
\]
This demonstrates that as $\mathcal{L}_{\text{gen}}$ decreases during training, the optimal noise scale $\sigma^\star$ shrinks accordingly. Furthermore, when $\sqrt{\mathcal{L}_{\text{gen}}}\approx \lambda$, $\sigma^\star$ approaches $0$, yielding nearly deterministic SID assignment and reducing the train--test mismatch.

\noindent \textbf{(2) Frequency-based Uncertainty Decay (FrqUD).}
Another uncertainty decay strategy is derived directly from codebook utilization in the implementation. The key idea is that \emph{high-frequency} codes indicate over-reuse (i.e., some codes are sampled too often), so we apply Gumbel noise to these ``hot'' codes to encourage exploration and improve code coverage. Conversely, \emph{low-frequency} codes are already under-used, so we keep their sampling stable and close to inference by using the standard deterministic assignment. Concretely, we compute an (EMA-smoothed) usage frequency for each code and threshold it by a ratio of the uniform baseline~\cite{gardner2006exponential}.

Let $f_i^{(e)}$ denote the usage frequency of code $i$ at epoch $e$, estimated either empirically within epoch $e$ or via an exponential moving average (EMA) over epochs:
\begin{equation}
f_i^{(e)} \leftarrow \beta f_i^{(e-1)} + (1-\beta)\hat f_i^{(e)}, \quad \beta\in[0,1),
\end{equation}
where $\hat f_i^{(e)}$ is the raw empirical frequency at epoch $e$ and $\sum_{i=1}^{K} f_i^{(e)} = 1$.
Under uniform usage, the average frequency is $\bar f = 1/K$. We define the hot-code threshold as:
\begin{equation}
\gamma = r\,\bar f = \frac{r}{K},
\end{equation}
where $r>0$ is a tunable \emph{threshold ratio} (e.g., $r=1.5$). We then identify over-used (hot) codes:
\begin{equation}
\mathcal{I}_{\text{high}}^{(e)} = \{\, i \mid f_i^{(e)} > \gamma \,\}, 
\qquad 
\mathcal{I}_{\text{low}}^{(e)} = \{1,\dots,K\}\setminus \mathcal{I}_{\text{high}}^{(e)}.
\end{equation}
For $i \in \mathcal{I}_{\text{high}}^{(e)}$, we update assignments with Gumbel noise. For $i \in \mathcal{I}_{\text{low}}^{(e)}$, we disable Gumbel noise, yielding a deterministic distribution:
\begin{equation}
y_i =
\begin{cases}
\displaystyle
\frac{\exp\!\left((\ell_i + g_i)/\tau\right)}
{\sum_{j=1}^{K}\exp\!\left((\ell_j +  g_j)/\tau\right)},
& i \in \mathcal{I}_{\text{high}}^{(e)}, \\[1.2em]
\displaystyle
\frac{\exp\!\left(\ell_i/\tau\right)}
{\sum_{j=1}^{K}\exp\!\left(\ell_j/\tau\right)},
& i \in \mathcal{I}_{\text{low}}^{(e)} .
\end{cases}
\end{equation}
where $\ell_i$ denotes the similarity logits produced by the encoder for code $i$, and $\tau$ is the temperature.

Both strategies aim to reduce Gumbel-noise-induced uncertainty later in
training while retaining the benefits of exploration earlier on, thereby improving training--inference objective alignment.

\subsection{Model Training}
We train the whole framework with joint optimization, so gradients flow through semantic ID assignment into the RQ-VAE codebook and the generative recommender. 
The training is primarily optimized using $\mathcal{L}{\text{gen}}$, the next-SID autoregressive generation loss (defined in Section~\ref{sec:preliminary}). To prevent the SID representations from drifting too far from the original data distribution, we additionally incorporate $\mathcal{L}{\text{vq}}$ and $\mathcal{L}{\text{recon}}$, the standard RQ-VAE losses, implemented following common formulations \cite{rajput2023recommender,lee2022autoregressive}.
\begin{table}[!htbp]
\centering
\setlength{\tabcolsep}{2pt}
\renewcommand{\arraystretch}{0.9}
\caption{Dataset statistics.}
\vspace{-0.15in}
\label{tab:dataset_stats}
\begin{tabular}{lrrrrr}
\hline
\multirow{2}{*}{Dataset} & \multirow{2}{*}{\#Users} & \multirow{2}{*}{\#Items} & \multirow{2}{*}{\#Interactions} & \multirow{2}{*}{Avg. Length} & \multirow{2}{*}{Sparsity} \\
\\
\hline
B-Shop & 22,363 & 12,101 & 198,502 & 8.88 & 0.9993 \\
I-Shop & 24,772 & 9,922 & 206,153 & 8.32 & 0.9992 \\
Yelp & 30,431 & 20,033 & 304,524 & 10.01 & 0.9995 \\
\hline
\end{tabular}
\vspace{-0.1in}
\end{table}
The main difference is that the RQ-VAE decoder is omitted from our training and inference pipeline: because our RQ-VAE's objective is to produce semantic IDs, we evaluate the reconstruction term directly on the RQ-quantized representations, rather than decoding them back to the original modality.
In practice, the model is initialized from a pretrained RQ-VAE, under which both $\mathcal{L}{\text{vq}}$ and $\mathcal{L}{\text{recon}}$ remain at a consistently smaller scale during training. As a result, parameter updates are primarily driven by $\mathcal{L}{\text{gen}}$, while the other terms stabilize optimization. We therefore optimize the joint objective $\mathcal{L}{\text{gen}}+\mathcal{L}{\text{vq}}+\mathcal{L}{\text{recon}}$, which corresponds to a unified training objective~\cite{liang2018variational}. Note that although this study addresses the code collapse problem, we intentionally did not incorporate a code-diversity loss~\cite{wang2024learnable}. 
This was done to maintain clarity and to isolate the effect of the differentiable SID alone.

\section{Experimental Setup}

\noindent \textbf{Datasets.} 
We conducted experiments on three public datasets commonly used in generative recommendation research, including two shopping datasets and one food-review dataset. For simplicity, we refer to them as B-Shop, I-Shop, and Yelp. B-Shop contains user interactions with cosmetic products, while I-Shop contains user interactions with music-related products. Dataset statistics are summarized in Table~\ref{tab:dataset_stats}. The B-Shop and I-Shop datasets consist of product review data accompanied by item descriptions; following common practice, we use the associated product text for semantic encoding. The Yelp dataset contains restaurant reviews, for which the restaurant name and business categories are incorporated as content features. For all three datasets, review texts are not used for item representation and are only employed to construct interaction records. We adopt the data processing pipeline proposed in \cite{wang2024learnable}, which employs a LLaMA-7B model to generate item content from the original titles and descriptions.
For each dataset, we treated a user's interaction history as the input sequence and the subsequent interaction as the prediction target. We followed standard preprocessing~\cite{kang2018self}: filtering users and items with fewer than five interactions~\cite{hua2023index}.

\noindent \textbf{Evaluation.} In this paper, we used Recall@10 and NDCG@10 (Normalized Discounted Cumulative Gain) as evaluation metrics. We followed the widely used leave-one-out evaluation protocol~\cite{he2017neural,wang2024learnable,liu2024generative}. All reported results were computed on the test set. Additionally, the predicted item was ranked against the full item set.

\noindent \textbf{Baselines.} We first adopt STE~\cite{bengio2013estimating} as our primary baseline, which represents a naive approach for learning differentiable Semantic IDs and is the default optimization choice in RQ-VAE-style discrete token learning. In addition, we compare with representative traditional collaborative-filtering and sequential recommendation methods, including MF~\cite{koren2009matrix}, LightGCN~\cite{he2020lightgcn}, Caser~\cite{tang2018personalized}, HGN~\cite{ma2019hierarchical}, SASRec~\cite{kang2018self}, and BERT4Rec~\cite{sun2019bert4rec}. We further include recent generative recommenders, namely BIGRec~\cite{bao2025bi}, P5~\cite{geng2022recommendation}, LETTER~\cite{wang2024learnable}, and ETEGRec~\cite{liu2024generative}. Among them, TIGER~\cite{rajput2023recommender} serves not only as a strong generative baseline but also as the main reference two-stage framework for our method following the mainstream approaches~\cite{liu2024generative,wang2024learnable};
All implemented methods were tuned for each dataset using the validation set. For a fair comparison, we report most results from~\cite{wang2024learnable} and ensure that the experimental configurations match so the results are directly comparable.

\noindent \textbf{Implementation Details.}
We followed the settings in \cite{liu2024generative}, and initialized the item tokenizer with a pretrained RQ-VAE, set the codebook size to $K=256$, and used code sequence length $m=3$ plus one code as the conflict code~\cite{rajput2023recommender}. Our generative recommender adopts a T5-style encoder--decoder backbone with 6 encoder and 6 decoder layers, hidden size 128. We trained the framework with AdamW with weight decay 0.05 and used early stopping based on validation performance. We tuned the Gumbel temperature $\tau$ from $\{0.5,1.0,2.0,4.0\}$ and set it to 2.0.  We tuned the learning rate for the recommender in $\{0.01,\,0.005,\,0.001,\,0.0005\}$
, the learning rate for the tokenizer in $\{10^{-4},10^{-5},5\times10^{-5},10^{-6}\}$
. We set $\beta$ to 0.25 in all experiments and found the method to be insensitive to small variations of this value. We tuned the hyperparameters $r$ and $\lambda$. Specifically, we perform a fine-grained grid search over $\lambda$  $\{1.0, 1.2, 1.4, 1.8, 2.0\}$, and tune $r$ over $\{1.0, 1.5, 2.0, 2.5, 3.0\}$.  All experiments were conducted on two H100 GPUs.

\begin{table*}[h]
\centering
\caption{Comparison of conventional two-stage generative recommendation and training with differentiable semantic IDs.}
\vspace{-0.1in}
\label{tab:performance_beauty_instruments}
\renewcommand{\arraystretch}{0.9}
\setlength{\tabcolsep}{4pt}
\begin{tabular}{l cccc cccc cccc} 
\toprule
\textbf{Model} & \multicolumn{4}{c}{\textbf{B-Shop}} & \multicolumn{4}{c}{\textbf{I-Shop}} & \multicolumn{4}{c}{\textbf{Yelp}} \\
\cmidrule(lr){2-5} \cmidrule(lr){6-9} \cmidrule(lr){10-13}
& \textbf{R@5} & \textbf{R@10} & \textbf{N@5} & \textbf{N@10} & \textbf{R@5} & \textbf{R@10} & \textbf{N@5} & \textbf{N@10} & \textbf{R@5} & \textbf{R@10} & \textbf{N@5} & \textbf{N@10} \\
\midrule
\textbf{Two-Stage}   & 0.0395 & 0.0610 & 0.0262 & 0.0331 & 0.0870 & 0.1058 & 0.0737 & 0.0797 & 0.0253 & 0.0407 & 0.0164 & 0.0213 \\
\hline
\textbf{STE} & 0.0076 &0.0134 &  0.0048 & 0.0067& 0.0450 & 0.0554 & 0.0330 &0.0360 & 0.0093 & 0.0147 & 0.0060 & 0.0077\\
\textbf{DIGER (FrqUD)} & 0.0440 & 0.0683 & \textbf{0.0294} & 0.0372 & \textbf{0.0915} & \textbf{0.1138} & \textbf{0.0772} & \textbf{0.0844} &0.0266 &0.0432 & 0.0173 & 0.0227\\
\textbf{DIGER (SDUD)}   & \textbf{0.0442} & 0.0657 & 0.0292 & 0.0361 & 0.0905 & 0.1124 & 0.0753 & 0.0823 & 0.0267 & \textbf{0.0439} & 0.0171 &0.0227 \\
\textbf{DIGER (SDUD+FrqUD)}& \textbf{0.0439} & \textbf{0.0696} & 0.0293 & \textbf{0.0376} & 0.0907 & 0.1127 & 0.0758
 & 0.0829 & \textbf{0.0273} &0.0437& \textbf{0.0175} & \textbf{0.0227} \\
\bottomrule
\end{tabular}
\vspace{-0.15in}
\end{table*}

\vspace{-0.05in}
\section{Experimental results}
In this section, we report our experimental results and address the following research questions:
\begin{itemize}
    \item \textbf{RQ1:} Does training with differentiable SID improve the performance of generative recommendation models?
    \item \textbf{RQ2:} How does the proposed DIGER compare with state-of-the-art methods?
    \item \textbf{RQ3:} How do different training strategies and key factors influence the final performance?
    \item \textbf{RQ4:} How do semantic ID assignments evolve during training, and how does decaying affect their stability, consistency, and utilization?
\end{itemize}
\vspace{-0.1in}

\begin{table}[h]
\centering
\setlength{\tabcolsep}{1.6pt}
\renewcommand{\arraystretch}{0.9}
\caption{Overall performance comparison between the baselines on generative recommender models on three datasets. The bold results highlight the better performance in comparing the backend models. R and N stand for Recall and NDCG. We denote "*" as the p-value < 0.05 over the second best baseline.}
\vspace{-0.1in}
\label{tab:performance_beauty_instruments_yelp_at10}
\setlength{\tabcolsep}{2pt}
\begin{tabular}{l cc cc cc} 
\toprule
\textbf{Model} & \multicolumn{2}{c}{\textbf{B-Shop}} & \multicolumn{2}{c}{\textbf{I-Shop}} & \multicolumn{2}{c}{\textbf{Yelp}} \\
\cmidrule(lr){2-3} \cmidrule(lr){4-5} \cmidrule(lr){6-7}
& \textbf{R@10} & \textbf{N@10} & \textbf{R@10} & \textbf{N@10} & \textbf{R@10} & \textbf{N@10} \\
\midrule
MF       & 0.0474 & 0.0191 & 0.0735 & 0.0412 & 0.0381 & 0.0190 \\
Caser    & 0.0347 & 0.0176 & 0.0710 & 0.0409 & 0.0263 & 0.0134 \\
HGN      & 0.0512 & 0.0266 & 0.1048 & 0.0774 & 0.0326 & 0.0159 \\
BERT4Rec & 0.0347 & 0.0170 & 0.0822 & 0.0608 & 0.0291 & 0.0159 \\
LightGCN & 0.0511 & 0.0260 & 0.1000 & 0.0728 & 0.0407 & 0.0207 \\
SASRec   & 0.0588 & 0.0313 & 0.0947 & 0.0690 & 0.0296 & 0.0152 \\
BigRec   & 0.0299 & 0.0198 & 0.0576 & 0.0491 & 0.0169 & 0.0142 \\
P5-SID & 0.0584 & 0.0335 & 0.0964 & 0.0730 & 0.0324 & 0.0170 \\
P5-CID   & 0.0597 & 0.0347 & 0.0987 & 0.0751 & 0.0347 & 0.0181 \\
LETTER   & \underline{0.0672} &  \underline{0.0364} & \underline{0.1122} &  \underline{0.0831} &  \underline{0.0426} & \textbf{0.0231} \\
TIGER    & 0.0610 & 0.0331 & 0.1058 & 0.0797 & 0.0407 & 0.0213 \\
ETEGRec  & 0.0615 & 0.0335 & 0.1106 & 0.0810 & 0.0415 & 0.0214 \\
\hline
DIGER (Ours)  & \textbf{0.0683$^{*}$} & \textbf{0.0372$^{*}$} & \textbf{0.1138$^{*}$} & \textbf{0.0844$^{*}$} & \textbf{0.0432} & \underline{0.0227} \\
\bottomrule
\end{tabular}
\vspace{-0.15in}
\end{table}

\subsection{Comparison with differentiable SID and conventional (RQ1)}
To address \textbf{RQ1}, we compare our differentiable semantic ID training framework DIGER with two representative alternatives: (i) a conventional two-stage generative recommendation pipeline (Two-Stage)~\cite{rajput2023recommender}, and (ii) a naive differentiable SID approach based on STE~\cite{bengio2013estimating}. We report experiments on three benchmarks (B-Shop, I-Shop, and Yelp). In addition, we evaluate several DIGER with different decaying strategies. The results are shown in  Table~\ref{tab:performance_beauty_instruments}. 

Table~\ref{tab:performance_beauty_instruments} shows that DIGER consistently improves over the conventional two-stage pipeline across all datasets and metrics. For instance, on B-Shop, DIGER raises R@10 from 0.0610 (Two-Stage) to 0.0657--0.0696 (DIGER), and N@10 from 0.0331 to 0.0361--0.0376. Similar trends are observed for I-Shop, where R@10 increases from 0.1058 to 0.1124--0.1138 and N@10 from 0.0797 to 0.0823--0.0844. On Yelp, DIGER also yields consistent gains (e.g., R@10 improves from 0.0407 to 0.0432--0.0439, and N@10 from 0.0213 to 0.0227). In contrast, the naive STE baseline performs poorly on all datasets (e.g., B-Shop R@10 drops to 0.0134), indicating severe optimization instability; empirically, this aligns with \emph{codebook collapse} behavior that prevents effective joint learning when discretization is handled in a purely straight-through manner.

Overall, these results provide empirical support for \textbf{RQ1}. Compared with a two-stage pipeline, DIGER unlocks additional gains by allowing the model to jointly optimize semantic ID construction and generation objectives. Moreover, the failure of STE highlights that \emph{a naive differentiable SID design is insufficient}; stable joint training requires the proposed DIGER to avoid collapse and to translate differentiability into measurable recommendation improvements.

\vspace{-0.1in}
\subsection{Comparison with SOTA baselines (RQ2)}
To address \textbf{RQ2}, we benchmark DIGER against a broad suite of strong sequential and generative recommenders, including classical collaborative filtering baselines, as well as recent generative recommendation methods. For fair comparison, we report R@10 and N@10 on the same three datasets (Table~\ref{tab:performance_beauty_instruments_yelp_at10}). Unless otherwise stated, our main DIGER result corresponds to the configuration with frequency-based uncertainty decay (\textbf{FrqUD}) as it provides strong and stable performance in practice (cf. Table~\ref{tab:performance_beauty_instruments}).

As shown in Table~\ref{tab:performance_beauty_instruments_yelp_at10}, DIGER achieves \textbf{state-of-the-art or best} performance on B-Shop and I-Shop
, outperforming all listed baselines. On Yelp, DIGER remains highly competitive: it achieves the best R@10 (0.0432), while its N@10 (0.0227) is comparable to the strongest baseline LETTER (0.0231). Notably, DIGER also consistently surpasses prior two-stage generative backbones such as TIGER across all datasets. For ETEGRec, its performance does not compare favorably with DIGER, despite requiring roughly twice as much training time. However, it outperforms the conventional approach, TIGER, slightly across all three datasets. This indicates that alternative distillation strategies for object alignment are effective to some extent; nonetheless, the remaining gap to DIGER is expected, as our method aligns representations through direct gradient-based joint optimization with differentiable semantic IDs.

Overall, we find that DIGER is competitive with, and often surpasses, state-of-the-art baselines on generative recommendation (\textbf{RQ2}). 
The small gap to LETTER on Yelp in N@10 is expected, as LETTER leverages additional collaborative signals (e.g., SASRec-derived embeddings) beyond pure text modeling, whereas DIGER is trained in a \emph{text-only} manner to better align with the core objective of generative recommendation.\footnote{Incorporating external collaborative embeddings can be viewed as a form of model ensembling or auxiliary information injection; exploring such hybridization is orthogonal to our focus and is therefore left as future work.}

\section{Ablation Study (RQ3)}
To answer \textbf{RQ3}, which investigates how different training strategies and key design choices influence the final performance, we conduct a comprehensive ablation study on DIGER. Our analysis focuses on three aspects: (i) the contribution of core components in differentiable semantic ID learning, (ii) the robustness to key uncertainty decay-related hyperparameters, and (iii) the effect of codebook capacity and semantic ID length. Unless otherwise stated, experiments are conducted on the B-Shop dataset, as similar trends are consistently observed on other datasets.

\noindent \textbf{Effect of core components.}
Table~\ref{tab:ablation_components_beauty} reports the impact of removing individual components. Removing the uncertainty decay strategies already causes a clear performance drop, highlighting their crucial role in facilitating the transition from exploration to exploitation during joint optimization. For the second part, we perform a consistent ablation based on DIGER-w/o UD, which retains Gumbel noise and the soft update. Removing the Gumbel noise causes a significant performance drop, highlighting its importance for effective exploration. The performance remains clearly better than STE, indicating that the soft update is superior to hard updates. Further removing the soft codebook update leads to a slight additional degradation. Furthermore, incorporating the design with temperature annealing yields similar but slightly inferior results, while substituting Gumbel noise with Gaussian noise causes a more pronounced degradation. This suggests that the asymmetric Gumbel distribution, which is designed to approximate categorical sampling by favoring higher-probability codes while preserving exploration, is better suited for discrete semantic ID learning than purely random perturbations.

\begin{table}[t]
\centering
\caption{Ablation of DIGER components on B-Shop Dataset. We omit other datasets due to page limitations.}
\vspace{-0.1in}
\renewcommand{\arraystretch}{0.8}
\label{tab:ablation_components_beauty}
\begin{tabular}{l|cccc}
\toprule
Model & R@5 & R@10 & N@5 & N@10 \\
\midrule
Two-Stage & 0.0395 & 0.0610 & 0.0262 & 0.0331 \\
Naive E2E (STE)  & 0.0076 & 0.0134 &   0.0048 & 0.0067 \\
DIGER-w/ UD & \textbf{0.0440} & \textbf{0.0683} & \textbf{0.0294} & \textbf{0.0372}\\
\midrule
DIGER-w/o UD & 0.0424 & 0.0679 & 0.0283 & 0.0365 \\
-w/o Gumbel Noise & 0.0168 & 0.0283 & 0.0104 & 0.0141 \\
-w/o soft update & 0.0422 & 0.0650 & 0.0281 & 0.0354 \\
\midrule
-w Gumbel tau annealing & 0.0419 & 0.0654 & 0.0273 & 0.0348\\
-w Gaussian Noise & 0.0389 & 0.0620 & 0.0253 & 0.0327\\
\bottomrule
\end{tabular}
\vspace{-0.1in}
\end{table}

\begin{figure}[t]
    \centering
    \includegraphics[width=\linewidth]{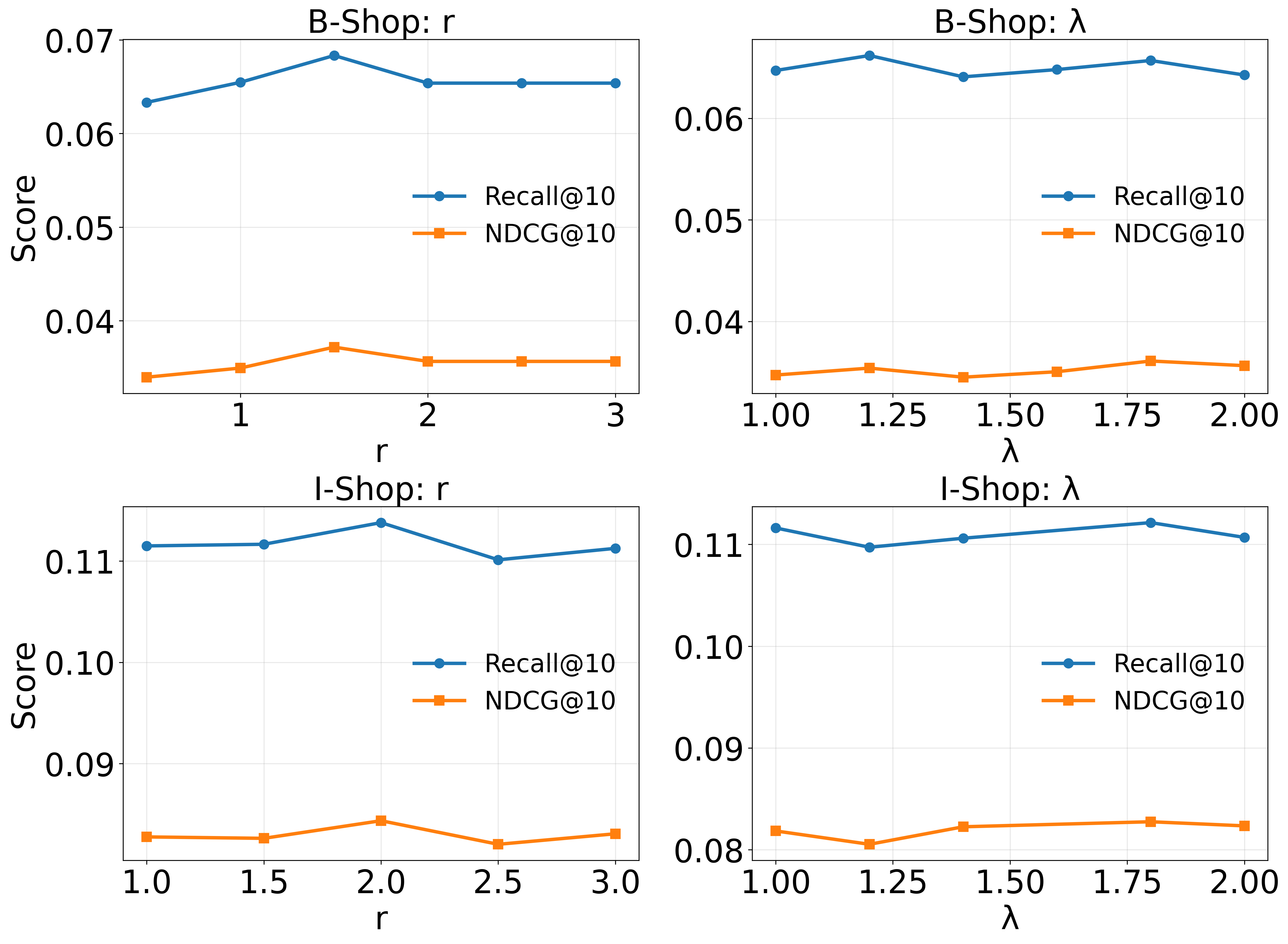}
    \vspace{-0.25in}
    \caption{Hyperparameter analysis on B-Shop and I-Shop. We vary the hot-threshold $r$ and $\lambda$, and report Recall@10 and NDCG@10.}
    \label{fig:hyperparam_2x2}
    \vspace{-0.2in}
\end{figure}

\begin{figure*}[t]
\centering
\includegraphics[width=0.95\textwidth]{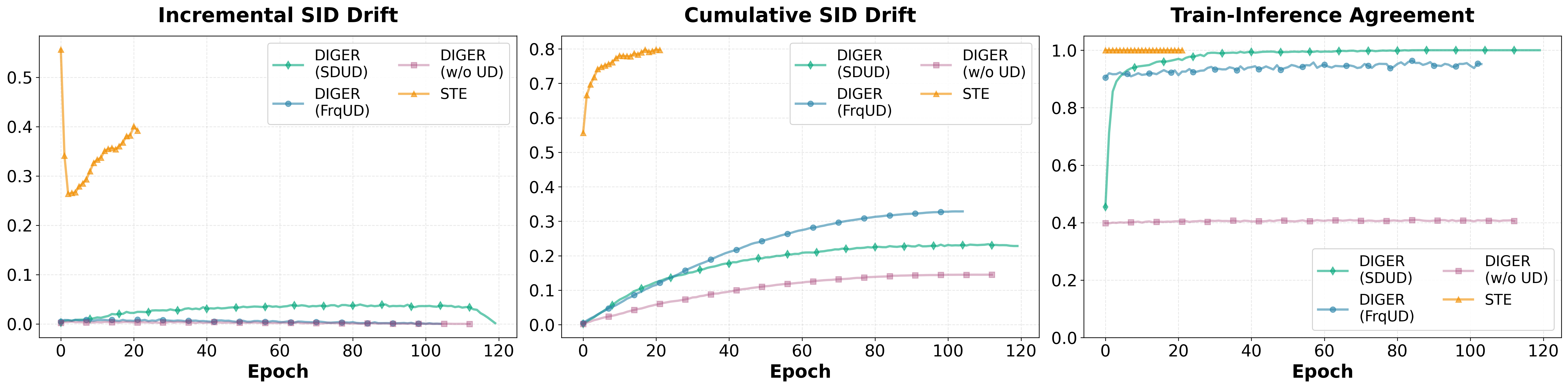}
\vspace{-0.15in}
\caption{Comparison of four semantic ID assignment methods over training epochs: (a) \textbf{Incremental SID Drift} measures the fraction of items whose semantic IDs change relative to the previous epoch; (b) \textbf{Cumulative SID Drift} measures the fraction of items whose semantic IDs differ from their initial values; (c) \textbf{Train--Inference Agreement} measures the fraction of items whose training-time sampled semantic IDs match the inference-time deterministic semantic IDs. Methods: \textbf{DIGER (w/o UD)}, \textbf{DIGER (FrqUD)}, \textbf{DIGER (SDUD)}, and \textbf{STE}.}
\vspace{-0.1in}
\label{fig:sid_drift_agreement}
\end{figure*}

\noindent \textbf{Sensitivity to Uncertainty Decay hyperparameters.}
We analyzed the sensitivity to $\lambda$, which controls when the uncertainty scale $\sigma$ approaches zero, and $r$, the hot-code threshold ratio in frequency-based decay. As shown in Figure~\ref{fig:hyperparam_2x2}, DIGER remains stable across a broad range of values. Very small $\lambda$ maintains excessive uncertainty late in training, while overly large values suppress early exploration; similarly, moderate $r$ effectively balances exploration and convergence, whereas extreme values weaken one of the two effects. Overall, performance variations remain limited within reasonable ranges, indicating that DIGER does not heavily rely on delicate hyperparameter tuning.

\noindent \textbf{Codebook capacity and semantic ID length.}
We further studied the effect of codebook size $K$ and semantic ID length $m$, with results shown in Table~\ref{tab:capacity_beauty}. Performance peaks at $K=256$, while both smaller and larger codebooks lead to inferior results, reflecting a trade-off between representational capacity and optimization stability. Reducing $m$ to 2 causes a substantial performance drop, indicating insufficient compositional capacity, whereas increasing $m$ to 4 yields only marginal gains. Overall, $K=256$ and $m=3$ provide a favorable balance between expressivity and stable training.

Overall, we find that DIGER's effectiveness stems from the joint design of controlled stochastic code selection, soft codebook updates, and uncertainty decay. The method is robust to hyperparameter choices and exhibits predictable behavior with respect to codebook capacity and semantic ID length \textbf{(RQ3)}.

\begin{table}[t]
\centering
\caption{Effect of codebook capacity ($K$) and SID length ($m$) on B-Shop Dataset.}
\vspace{-0.1in}
\renewcommand{\arraystretch}{0.75}
\label{tab:capacity_beauty}
\setlength{\tabcolsep}{10pt}
\begin{tabular}{c c|cccc}
\toprule
\multirow{2}{*}{$K$} & \multirow{2}{*}{$m$} & \multirow{2}{*}{R@5} & \multirow{2}{*}{R@10} & \multirow{2}{*}{N@5} & \multirow{2}{*}{N@10} \\
&&&&&\\
\midrule
128 & 3 & 0.0419 & 0.0655 & 0.0276 & 0.0352 \\
256 & 3 & \textbf{0.0440} & \textbf{0.0683} & \textbf{0.0294} & \textbf{0.0372} \\
512 & 3 & 0.0434 & 0.0660 & 0.0286 & 0.0359 \\
\midrule
256 & 2 & 0.0362 & 0.0566 & 0.0236 & 0.0301 \\
256 & 3 & \textbf{0.0440} & \textbf{0.0683} & 0.0294 & 0.0372 \\
256 & 4 & 0.0437 & 0.0673 & \textbf{0.0297} & \textbf{0.0373} \\
\bottomrule
\end{tabular}
\vspace{-0.1in}
\end{table}

\vspace{-0.1in}
\section{Semantic ID Dynamics and Stability (RQ4)}
To answer \textbf{RQ4}, we analyzed how semantic ID (SID) assignments evolve during joint training and how uncertainty decay influences their stability, consistency, and codebook utilization. We focused on three aspects: SID drift over training, training--inference agreement, and the code usage distribution. Representative Results are reported on the B-Shop dataset. Other datasets are omitted due to page limit.

\noindent \textbf{SID drift during joint training.}
To characterize how semantic IDs evolve, we measure both incremental SID drift, which captures epoch-to-epoch assignment changes, and cumulative SID drift, which measures deviations from the initial pretrained assignments. These two metrics reflect different aspects of SID dynamics and are not directly convertible. As shown in Figure~\ref{fig:sid_drift_agreement}, naive STE exhibits much larger and more abrupt incremental drift than all DIGER variants, particularly in early training, where a large fraction of items change their SIDs within a single epoch. This overly deterministic behavior leads to unstable optimization and is consistent with the observed collapse phenomena. In contrast, all DIGER variants evolve more gradually. While DIGER with uncertainty decay, especially FrqUD, shows larger cumulative drift, it remains below 40\%, indicating that most semantic IDs are preserved and only selectively refined. DIGER without uncertainty decay exhibits the smallest cumulative drift, below 15\%, suggesting limited adaptation. Importantly, variants with uncertainty decay consistently achieve better performance, supporting the benefit of allowing controlled SID changes during later-stage exploitation.

\begin{figure}[t]
\centering
\includegraphics[width=0.9\linewidth]{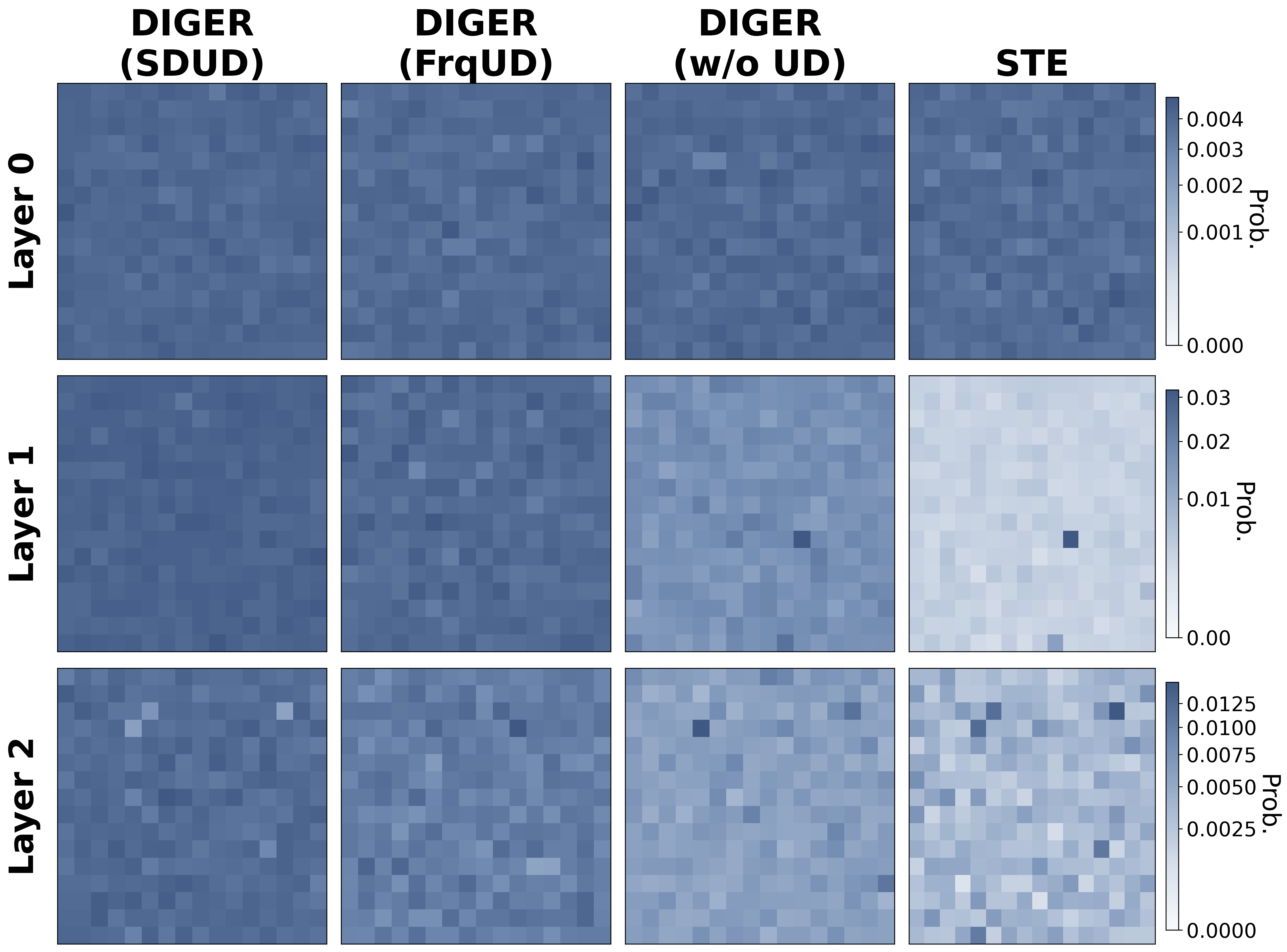}
\vspace{-0.1in}
\caption{Code usage distribution at the best checkpoint for four methods across three quantization layers. Each 16×16 grid visualizes the 256 codebook entries, with color intensity indicating usage probability.}
\label{fig:code_usage_heatmap}
\vspace{-0.2in}
\end{figure}

\noindent \textbf{Training--inference agreement.}
Because Gumbel noise introduces stochasticity during training, we further examine the agreement between training-time sampled semantic IDs and inference-time deterministic (argmax) semantic IDs. As shown in Figure~\ref{fig:sid_drift_agreement}, STE achieves perfect agreement by construction, but this rigid determinism coincides with severe instability in semantic ID learning and poor recommendation performance. DIGER without uncertainty decay exhibits persistently low agreement, indicating a large mismatch between stochastic training assignments and deterministic inference behavior. In contrast, DIGER with uncertainty decay rapidly attains and maintains high agreement after a brief early phase. Notably, FrqUD enforces frequency-aware uncertainty from the beginning by applying stochastic sampling only to overused codes, leading to consistently high agreement throughout training, while SDUD gradually increases agreement by progressively reducing the Gumbel noise's standard deviation $\sigma$. Together, these results show that uncertainty decay effectively reconciles stochastic exploration with inference-time determinism, enabling stable and reliable joint optimization.

\noindent \textbf{Code usage distribution.}
We further inspect the code usage distribution at the best checkpoint across all quantization layers, as shown in Figure~\ref{fig:code_usage_heatmap}. DIGER with uncertainty decay exhibits the most balanced utilization, where the majority of codes are frequently used across layers. DIGER without uncertainty decay shows moderate imbalance, while STE suffers from severe code collapse, with many codes rarely or never selected, especially in deeper layers. Notably, all methods exhibit relatively uniform usage at the first quantization layer, but differences become pronounced at higher layers, where uncertainty decay plays a critical role in preventing collapse and encouraging effective use of the codebooks.

Overall, SIDs evolve throughout training, and uncertainty decay makes this evolution stable and effective \textbf{(RQ4)}. It aligns training and inference assignments, prevents abrupt SID drift and code collapse, and promotes balanced code usage, enabling controlled refinement of semantic IDs with reliable inference behavior.

\vspace{-0.08in}
\section{Conclusions}

This work revisits a key limitation in generative recommendation systems: semantic IDs are typically learned for reconstruction and then frozen, creating a mismatch between semantic indexing and the recommendation objective. Differentiable semantic IDs are therefore required to bridge this gap. However, the naive solution based on straight-through estimators (STE) suffers from severe code collapse issues. To address this issue, DIGER introduces differentiable semantic ID exploratory learning (DRIL), which supports effective SID exploration, together with two uncertainty decay strategies that guide a smooth transition from exploration to exploitation. DIGER stabilizes training, prevents semantic ID collapse, improves codebook utilization, and achieves competitive results on public benchmarks.

Looking ahead, DIGER opens several directions for future work. Differentiable semantic ID can be extended beyond item-side representations, such as learning user-side or interaction-level discrete structures. More generally, exploring alternative estimators and optimization strategies for discrete latent variables may further improve stability and expressivity. Finally, integrating differentiable semantic IDs with richer collaborative signals or large language model–based recommendation frameworks remains a promising avenue for advancing generative recommendation for jointly training SID and recommender.

\appendix
\vspace{-0.08in}
\section{Theoretical Analysis}
\label{sec:appendix}
\paragraph{Theoretical justification.}
Let $\theta$ be the recommender parameters and $\phi$ be the semantic indexing parameters (codebooks).
Let the target (recommendation) objective be $\mathcal{L}_{\mathrm{rec}}(\theta,\phi)$.

\begin{theorem}[Two-stage freezing as restricted minimization over $\phi$]\label{theorema1}
Assume the two-stage scheme first selects $\phi$ from a subset $\mathcal{A}\subseteq\Phi$ (e.g., $\mathcal{A}=\arg\min_\phi \mathcal{L}_{\mathrm{aux}}(\phi)$ or the set of outputs of a pretraining routine), and then optimizes $\theta$ with $\phi$ fixed.
Define the value function $g(\phi):=\inf_{\theta\in\Theta}\mathcal{L}_{\mathrm{rec}}(\theta,\phi)$.
Then the best two-stage performance satisfies
\(
J_{\mathrm{2st}}^\star=\inf_{\phi\in\mathcal{A}} g(\phi)\),
\(
J_{\mathrm{e2e}}=\inf_{\phi\in\Phi} g(\phi),
\)
and hence $J_{\mathrm{e2e}}\le J_{\mathrm{2st}}^\star$.
Moreover, if $g$ is continuous and $\mathcal{A}$ is compact and does not contain any global minimizer of $g$ over $\Phi$, then $J_{\mathrm{e2e}}<J_{\mathrm{2st}}^\star$.
\end{theorem}

\begin{proof}
The identities follow from definition of $g$ and the two-stage protocol. Since $\mathcal{A}\subseteq\Phi$, we have
$\inf_{\phi\in\Phi}g(\phi)\le \inf_{\phi\in\mathcal{A}}g(\phi)$.
The strict inequality under the stated conditions follows by the extreme value theorem and the assumption that $\mathcal{A}$ excludes all minimizers of $g$.
\end{proof}

\begin{theorem}[Arbitrary suboptimality under objective mismatch]\label{theorema2}
Assume there exists $\phi^\star \in \arg\min_{\phi}\min_{\theta}\mathcal{L}_{\mathrm{rec}}(\theta,\phi)$ and a (possibly different) $\phi_{\mathrm{aux}} \in \arg\min_{\phi}\mathcal{L}_{\mathrm{aux}}(\phi)$ with $\phi_{\mathrm{aux}}\neq \phi^\star$.
Then there exist continuous losses $\mathcal{L}_{\mathrm{rec}}$ and $\mathcal{L}_{\mathrm{aux}}$ such that the optimal two-stage solution is strictly worse than the jointly optimized optimum, and the gap can be made arbitrarily large:
\[
\inf_{\theta}\mathcal{L}_{\mathrm{rec}}(\theta,\phi_{\mathrm{aux}}) - \min_{\theta,\phi}\mathcal{L}_{\mathrm{rec}}(\theta,\phi) \ge M,
\quad \forall M>0.
\]
\end{theorem}
\begin{proof}[Proof (constructive)]
Take $\phi \in \mathbb{R}$ and define $\mathcal{L}_{\mathrm{aux}}(\phi) = (\phi-1)^2$, so $\phi_{\mathrm{aux}}=1$.
Let $\mathcal{L}_{\mathrm{rec}}(\theta,\phi) = (\phi-0)^2 + \theta^2$, so the jointly optimized optimum is attained at $(\theta^\star,\phi^\star)=(0,0)$ with value $0$.
The best two-stage value (freezing $\phi_{\mathrm{aux}}=1$) is $\inf_{\theta}\mathcal{L}_{\mathrm{rec}}(\theta,1)=1$.
Scaling $\mathcal{L}_{\mathrm{rec}}$ by $M$ yields gap $M$.
\end{proof}

\begin{remark}\label{remark:jo}
By Theorem~\ref{theorema1}, two-stage training optimizes the recommendation loss over a restricted subset of indexing parameters. Theorem~\ref{theorema2} shows that if this subset excludes the joint optimum due to objective mismatch, the resulting suboptimality can be arbitrarily large, even under exact optimization.
\end{remark}

\paragraph{Noise/entropy and codebook utilization.}
Let $q \in \Delta^{K-1}$ be the marginal code-usage distribution (e.g., $q_i=\mathbb{E}_{v}[y_i(v)]$), and define the effective number of used codes as $\mathrm{Eff}(q):=\exp(H(q))$ where $H(q)=-\sum_{i=1}^K q_i\log q_i$.

\begin{theorem}[Entropy regularization maximizes effective code usage]
For any $K\ge 2$, $\mathrm{Eff}(q)$ is maximized uniquely at the uniform distribution $q_i=1/K$, achieving $\mathrm{Eff}(q)=K$.
Equivalently, $H(q)\le \log K$ with equality iff $q$ is uniform.
\end{theorem}
\begin{proof}
This is the standard maximal-entropy property on the probability simplex: $H(q)$ is strictly concave on $\Delta^{K-1}$ and attains its maximum at the uniform distribution.
\end{proof}

\begin{remark}\label{remark:entropy}
Therefore, adding an entropy bonus that increases the entropy of assignments (or of the induced marginal $q$) has a direct, formally justified effect: it pushes the solution toward higher $\mathrm{Eff}(q)$, i.e., better codebook utilization. This is a property of the objective itself, independent of a particular optimizer.
\end{remark}

\vspace{-0.15in}

\section{Derivation of the optimal $\sigma$.}
\label{sec:derive_of_sigma}
Let
\begin{equation}
L_\sigma(\sigma)=\frac{L_{\mathrm{gen}}}{2(\sigma+\lambda)^2}+\log(\sigma+\lambda),
\end{equation}
where $L_{\mathrm{gen}}$ and $\lambda$ are treated as constants w.r.t.\ $\sigma$. Differentiating w.r.t.\ $\sigma$ yields
\begin{equation}
\frac{\partial L_\sigma}{\partial \sigma}
=\frac{L_{\mathrm{gen}}}{2}\cdot(-2)(\sigma+\lambda)^{-3}+(\sigma+\lambda)^{-1}
=-\frac{L_{\mathrm{gen}}}{(\sigma+\lambda)^3}+\frac{1}{\sigma+\lambda}.
\end{equation}
Setting $\partial L_\sigma/\partial \sigma=0$ and multiplying both sides by $(\sigma+\lambda)^3$ gives
\begin{equation}
-(L_{\mathrm{gen}})+(\sigma+\lambda)^2=0
\quad\Rightarrow\quad
(\sigma+\lambda)^2=L_{\mathrm{gen}}.
\end{equation}
The domain of $\log(\sigma+\lambda)$ requires $\sigma+\lambda>0$, hence we take the positive root:
\begin{equation}
\sigma+\lambda=\sqrt{L_{\mathrm{gen}}}
\quad\Rightarrow\quad
\sigma^\star=\sqrt{L_{\mathrm{gen}}}-\lambda.
\end{equation}
Finally, if we additionally enforce the practical constraint $\sigma\ge 0$, the closed-form solution becomes
\begin{equation}
\sigma^\star=\max\{0,\sqrt{L_{\mathrm{gen}}}-\lambda\}.
\end{equation}

\bibliographystyle{ACM-Reference-Format}
\bibliography{sample-base}

@inproceedings{fu2024exploring,
  title={Exploring adapter-based transfer learning for recommender systems: Empirical studies and practical insights},
  author={Fu, Junchen and Yuan, Fajie and Song, Yu and Yuan, Zheng and Cheng, Mingyue and Cheng, Shenghui and Zhang, Jiaqi and Wang, Jie and Pan, Yunzhu},
  booktitle={Proceedings of the 17th ACM international conference on web search and data mining},
  pages={208--217},
  year={2024}
}

@inproceedings{fu2024iisan,
  title={IISAN: Efficiently adapting multimodal representation for sequential recommendation with decoupled PEFT},
  author={Fu, Junchen and Ge, Xuri and Xin, Xin and Karatzoglou, Alexandros and Arapakis, Ioannis and Wang, Jie and Jose, Joemon M},
  booktitle={Proceedings of the 47th International ACM SIGIR Conference on Research and Development in Information Retrieval},
  pages={687--697},
  year={2024}
}

@article{fu2025efficient,
  title={Efficient and effective adaptation of multimodal foundation models in sequential recommendation},
  author={Fu, Junchen and Ge, Xuri and Xin, Xin and Karatzoglou, Alexandros and Arapakis, Ioannis and Zheng, Kaiwen and Ni, Yongxin and Joemon, Joemon M Jose},
  journal={IEEE Transactions on Knowledge and Data Engineering},
  year={2025},
  publisher={IEEE}
}

@article{zhao2026unifying,
  title={Unifying Search and Recommendation in LLMs via Gradient Multi-Subspace Tuning},
  author={Zhao, Jujia and Wang, Zihan and Pan, Shuaiqun and Verberne, Suzan and Ren, Zhaochun},
  journal={arXiv preprint arXiv:2601.09496
        
        
        
        },
  year={2026}
}

@inproceedings{sun2019bert4rec,
  title={BERT4Rec: Sequential recommendation with bidirectional encoder representations from transformer},
  author={Sun, Fei and Liu, Jun and Wu, Jian and Pei, Changhua and Lin, Xiao and Ou, Wenwu and Jiang, Peng},
  booktitle={Proceedings of the 28th ACM international conference on information and knowledge management},
  pages={1441--1450},
  year={2019}
}

@inproceedings{liang2018variational,
  title={Variational autoencoders for collaborative filtering},
  author={Liang, Dawen and Krishnan, Rahul G and Hoffman, Matthew D and Jebara, Tony},
  booktitle={Proceedings of the 2018 world wide web conference},
  pages={689--698},
  year={2018}
}

@article{graves2011practical,
  title={Practical variational inference for neural networks},
  author={Graves, Alex},
  journal={Advances in neural information processing systems},
  volume={24},
  year={2011}
}

@inproceedings{liu2024once,
  title={Once: Boosting content-based recommendation with both open-and closed-source large language models},
  author={Liu, Qijiong and Chen, Nuo and Sakai, Tetsuya and Wu, Xiao-Ming},
  booktitle={Proceedings of the 17th ACM International Conference on Web Search and Data Mining},
  pages={452--461},
  year={2024}
}

@article{gupta2006interplay,
  title={The interplay between exploration and exploitation},
  author={Gupta, Anil K and Smith, Ken G and Shalley, Christina E},
  journal={Academy of management journal},
  volume={49},
  number={4},
  pages={693--706},
  year={2006},
  publisher={Academy of Management Briarcliff Manor, NY 10510}
}

@article{lu2020mean,
  title={Mean-field approximation to Gaussian-softmax integral with application to uncertainty estimation},
  author={Lu, Zhiyun and Ie, Eugene and Sha, Fei},
  journal={arXiv preprint arXiv:2006.07584
        
        
        
        
        
        
        
        
        
        
        
        
        
        },
  year={2020}
}

@article{gardner2006exponential,
  title={Exponential smoothing: The state of the art—Part II},
  author={Gardner Jr, Everette S},
  journal={International journal of forecasting},
  volume={22},
  number={4},
  pages={637--666},
  year={2006},
  publisher={Elsevier}
}

@inproceedings{huh2023straightening,
  title={Straightening out the straight-through estimator: Overcoming optimization challenges in vector quantized networks},
  author={Huh, Minyoung and Cheung, Brian and Agrawal, Pulkit and Isola, Phillip},
  booktitle={International Conference on Machine Learning},
  pages={14096--14113},
  year={2023},
  organization={PMLR}
}

@article{zheng2025pre,
  title={Pre-training Generative Recommender with Multi-Identifier Item Tokenization},
  author={Zheng, Bowen and Liu, Enze and Chen, Zhongfu and Ma, Zhongrui and Wang, Yue and Zhao, Wayne Xin and Wen, Ji-Rong},
  journal={arXiv preprint arXiv:2504.04400
        
        
        
        
        
        
        
        
        
        
        
        
        
        
        
        
        
        
        
        },
  year={2025}
}

@inproceedings{zhang2025c2t,
  title={C2T-ID: Converting Semantic Codebooks to Textual Document Identifiers for Generative Search},
  author={Zhang, Yingchen and Zhang, Ruqing and Guo, Jiafeng and Peng, Wenjun and Li, Sen and Lv, Fuyu and Cheng, Xueqi},
  booktitle={Proceedings of the 2025 Annual International ACM SIGIR Conference on Research and Development in Information Retrieval in the Asia Pacific Region},
  pages={331--336},
  year={2025}
}

@article{liu2025discrec,
  title={DiscRec: Disentangled Semantic-Collaborative Modeling for Generative Recommendation},
  author={Liu, Chang and Bai, Yimeng and Zhao, Xiaoyan and Zhang, Yang and Feng, Fuli and Rong, Wenge},
  journal={arXiv preprint arXiv:2506.15576
        
        
        
        
        
        
        
        
        
        },
  year={2025}
}

@inproceedings{wang2025generative2,
  title={Generative next poi recommendation with semantic id},
  author={Wang, Dongsheng and Huang, Yuxi and Gao, Shen and Wang, Yifan and Huang, Chengrui and Shang, Shuo},
  booktitle={Proceedings of the 31st ACM SIGKDD Conference on Knowledge Discovery and Data Mining V. 2},
  pages={2904--2914},
  year={2025}
}

@article{zhai2025simple,
  title={A Simple Contrastive Framework Of Item Tokenization For Generative Recommendation},
  author={Zhai, Penglong and Yuan, Yifang and Di, Fanyi and Li, Jie and Liu, Yue and Li, Chen and Huang, Jie and Wang, Sicong and Xu, Yao and Li, Xin},
  journal={arXiv preprint arXiv:2506.16683
        
        
        
        
        
        
        
        
        
        
        
        },
  year={2025}
}

@inproceedings{li2025dimerec,
  title={DimeRec: a unified framework for enhanced sequential recommendation via generative diffusion models},
  author={Li, Wuchao and Huang, Rui and Zhao, Haijun and Liu, Chi and Zheng, Kai and Liu, Qi and Mou, Na and Zhou, Guorui and Lian, Defu and Song, Yang and others},
  booktitle={Proceedings of the Eighteenth ACM International Conference on Web Search and Data Mining},
  pages={726--734},
  year={2025}
}

@inproceedings{zhang2025killing,
  title={Killing two birds with one stone: Unifying retrieval and ranking with a single generative recommendation model},
  author={Zhang, Luankang and Song, Kenan and Lee, Yi Quan and Guo, Wei and Wang, Hao and Li, Yawen and Guo, Huifeng and Liu, Yong and Lian, Defu and Chen, Enhong},
  booktitle={Proceedings of the 48th International ACM SIGIR Conference on Research and Development in Information Retrieval},
  pages={2224--2234},
  year={2025}
}

@article{zhao2025unifying,
  title={Unifying Search and Recommendation: A Generative Paradigm Inspired by Information Theory},
  author={Zhao, Jujia and Wang, Wenjie and Xu, Chen and Chen, Xiuying and Ren, Zhaochun and Verberne, Suzan},
  journal={arXiv preprint arXiv:2504.06714
        
        
        
        
        
        
        
        
        
        
        
        
        
        },
  year={2025}
}

@inproceedings{wang2025empowering,
  title={Empowering large language model for sequential recommendation via multimodal embeddings and semantic ids},
  author={Wang, Yuhao and Pan, Junwei and Li, Xinhang and Wang, Maolin and Wang, Yuan and Liu, Yue and Liu, Dapeng and Jiang, Jie and Zhao, Xiangyu},
  booktitle={Proceedings of the 34th ACM International Conference on Information and Knowledge Management},
  pages={3209--3219},
  year={2025}
}

@article{liu2025onerec,
  title={Onerec-think: In-text reasoning for generative recommendation},
  author={Liu, Zhanyu and Wang, Shiyao and Wang, Xingmei and Zhang, Rongzhou and Deng, Jiaxin and Bao, Honghui and Zhang, Jinghao and Li, Wuchao and Zheng, Pengfei and Wu, Xiangyu and others},
  journal={arXiv preprint arXiv:2510.11639
        
        
        
        
        
        
        
        
        
        },
  year={2025}
}

@article{li2025survey,
  title={A Survey of Generative Recommendation from a Tri-Decoupled Perspective: Tokenization, Architecture, and Optimization},
  author={Li, Xiaopeng and Chen, Bo and She, Junda and Cao, Shiteng and Wang, You and Jia, Qinlin and He, Haiying and Zhou, Zheli and Liu, Zhao and Liu, Ji and others},
  year={2025},
  publisher={Preprints}
}

@inproceedings{hou2025generative,
  title={Generative Recommendation Models: Progress and Directions},
  author={Hou, Yupeng and Zhang, An and Sheng, Leheng and Yang, Zhengyi and Wang, Xiang and Chua, Tat-Seng and McAuley, Julian},
  booktitle={Companion Proceedings of the ACM on Web Conference 2025},
  pages={13--16},
  year={2025}
}

@inproceedings{han2025mtgr,
  title={Mtgr: Industrial-scale generative recommendation framework in meituan},
  author={Han, Ruidong and Yin, Bin and Chen, Shangyu and Jiang, He and Jiang, Fei and Li, Xiang and Ma, Chi and Huang, Mincong and Li, Xiaoguang and Jing, Chunzhen and others},
  booktitle={Proceedings of the 34th ACM International Conference on Information and Knowledge Management},
  pages={5731--5738},
  year={2025}
}

@inproceedings{wang2025generative,
  title={Generative large recommendation models: Emerging trends in llms for recommendation},
  author={Wang, Hao and Guo, Wei and Zhang, Luankang and Chin, Jin Yao and Ye, Yufei and Guo, Huifeng and Liu, Yong and Lian, Defu and Tang, Ruiming and Chen, Enhong},
  booktitle={Companion Proceedings of the ACM on Web Conference 2025},
  pages={49--52},
  year={2025}
}

@inproceedings{hou2025generating,
  title={Generating long semantic ids in parallel for recommendation},
  author={Hou, Yupeng and Li, Jiacheng and Shin, Ashley and Jeon, Jinsung and Santhanam, Abhishek and Shao, Wei and Hassani, Kaveh and Yao, Ning and McAuley, Julian},
  booktitle={Proceedings of the 31st ACM SIGKDD Conference on Knowledge Discovery and Data Mining V. 2},
  pages={956--966},
  year={2025}
}

@inproceedings{ye2025harnessing,
  title={Harnessing multimodal large language models for multimodal sequential recommendation},
  author={Ye, Yuyang and Zheng, Zhi and Shen, Yishan and Wang, Tianshu and Zhang, Hengruo and Zhu, Peijun and Yu, Runlong and Zhang, Kai and Xiong, Hui},
  booktitle={Proceedings of the AAAI Conference on Artificial Intelligence},
  volume={39},
  number={12},
  pages={13069--13077},
  year={2025}
}

@inproceedings{hua2023index,
  title={How to index item ids for recommendation foundation models},
  author={Hua, Wenyue and Xu, Shuyuan and Ge, Yingqiang and Zhang, Yongfeng},
  booktitle={Proceedings of the Annual International ACM SIGIR Conference on Research and Development in Information Retrieval in the Asia Pacific Region},
  pages={195--204},
  year={2023}
}

@inproceedings{he2017neural,
  title={Neural collaborative filtering},
  author={He, Xiangnan and Liao, Lizi and Zhang, Hanwang and Nie, Liqiang and Hu, Xia and Chua, Tat-Seng},
  booktitle={Proceedings of the 26th international conference on world wide web},
  pages={173--182},
  year={2017}
}

@article{bao2025bi,
  title={A bi-step grounding paradigm for large language models in recommendation systems},
  author={Bao, Keqin and Zhang, Jizhi and Wang, Wenjie and Zhang, Yang and Yang, Zhengyi and Luo, Yanchen and Chen, Chong and Feng, Fuli and Tian, Qi},
  journal={ACM Transactions on Recommender Systems},
  volume={3},
  number={4},
  pages={1--27},
  year={2025},
  publisher={ACM New York, NY}
}

@article{koren2009matrix,
  title={Matrix factorization techniques for recommender systems},
  author={Koren, Yehuda and Bell, Robert and Volinsky, Chris},
  journal={Computer},
  volume={42},
  number={8},
  pages={30--37},
  year={2009},
  publisher={IEEE}
}

@article{takida2022sq,
  title={Sq-vae: Variational bayes on discrete representation with self-annealed stochastic quantization},
  author={Takida, Yuhta and Shibuya, Takashi and Liao, WeiHsiang and Lai, Chieh-Hsin and Ohmura, Junki and Uesaka, Toshimitsu and Murata, Naoki and Takahashi, Shusuke and Kumakura, Toshiyuki and Mitsufuji, Yuki},
  journal={arXiv preprint arXiv:2205.07547
        
        
        
        
        
        
        
        
        
        
        
        
        
        },
  year={2022}
}

@inproceedings{lee2022autoregressive,
  title={Autoregressive image generation using residual quantization},
  author={Lee, Doyup and Kim, Chiheon and Kim, Saehoon and Cho, Minsu and Han, Wook-Shin},
  booktitle={Proceedings of the IEEE/CVF conference on computer vision and pattern recognition},
  pages={11523--11532},
  year={2022}
}

@article{bengio2013estimating,
  title={Estimating or propagating gradients through stochastic neurons for conditional computation},
  author={Bengio, Yoshua and L{\'e}onard, Nicholas and Courville, Aaron},
  journal={arXiv preprint arXiv:1308.3432},
  year={2013}
}

@article{hou2025survey,
  title={A survey on generative recommendation: Data, model, and tasks},
  author={Hou, Min and Wu, Le and Liao, Yuxin and Yang, Yonghui and Zhang, Zhen and Zheng, Changlong and Wu, Han and Hong, Richang},
  journal={arXiv preprint arXiv:2510.27157},
  year={2025}
}

@article{tay2022transformer,
  title={Transformer memory as a differentiable search index},
  author={Tay, Yi and Tran, Vinh and Dehghani, Mostafa and Ni, Jianmo and Bahri, Dara and Mehta, Harsh and Qin, Zhen and Hui, Kai and Zhao, Zhe and Gupta, Jai and others},
  journal={Advances in Neural Information Processing Systems},
  volume={35},
  pages={21831--21843},
  year={2022}
}

@inproceedings{he2020lightgcn,
  title={Lightgcn: Simplifying and powering graph convolution network for recommendation},
  author={He, Xiangnan and Deng, Kuan and Wang, Xiang and Li, Yan and Zhang, Yongdong and Wang, Meng},
  booktitle={Proceedings of the 43rd International ACM SIGIR conference on research and development in Information Retrieval},
  pages={639--648},
  year={2020}
}

@inproceedings{ma2019hierarchical,
  title={Hierarchical gating networks for sequential recommendation},
  author={Ma, Chen and Kang, Peng and Liu, Xue},
  booktitle={Proceedings of the 25th ACM SIGKDD international conference on knowledge discovery \& data mining},
  pages={825--833},
  year={2019}
}

@inproceedings{tang2018personalized,
  title={Personalized top-n sequential recommendation via convolutional sequence embedding},
  author={Tang, Jiaxi and Wang, Ke},
  booktitle={Proceedings of the eleventh ACM international conference on web search and data mining},
  pages={565--573},
  year={2018}
}

@inproceedings{kang2018self,
  title={Self-attentive sequential recommendation},
  author={Kang, Wang-Cheng and McAuley, Julian},
  booktitle={2018 IEEE international conference on data mining (ICDM)},
  pages={197--206},
  year={2018},
  organization={IEEE}
}

@inproceedings{li2023exploring,
author = {Li, Ruyu and Deng, Wenhao and Cheng, Yu and Yuan, Zheng and Zhang, Jiaqi and Yuan, Fajie},
title = {Exploring the Upper Limits of Text-Based Collaborative Filtering Using Large Language Models: Discoveries and Insights},
year = {2025},
isbn = {9798400720406},
publisher = {Association for Computing Machinery},
address = {New York, NY, USA},
url = {https://doi.org/10.1145/3746252.3761429},
doi = {10.1145/3746252.3761429},
booktitle = {Proceedings of the 34th ACM International Conference on Information and Knowledge Management},
pages = {1643–1653},
numpages = {11},
location = {Seoul, Republic of Korea},
series = {CIKM '25}
}

@inproceedings{wu2021empowering,
  title={Empowering news recommendation with pre-trained language models},
  author={Wu, Chuhan and Wu, Fangzhao and Qi, Tao and Huang, Yongfeng},
  booktitle={Proceedings of the 44th international ACM SIGIR conference on research and development in information retrieval},
  pages={1652--1656},
  year={2021}
}

@inproceedings{elsayed2022end,
  title={End-to-end image-based fashion recommendation},
  author={Elsayed, Shereen and Brinkmeyer, Lukas and Schmidt-Thieme, Lars},
  booktitle={Workshop on Recommender Systems in Fashion and Retail},
  pages={109--119},
  year={2022},
  organization={Springer}
}

@inproceedings{yuan2023go,
  title={Where to go next for recommender systems? id-vs. modality-based recommender models revisited},
  author={Yuan, Zheng and Yuan, Fajie and Song, Yu and Li, Youhua and Fu, Junchen and Yang, Fei and Pan, Yunzhu and Ni, Yongxin},
  booktitle={Proceedings of the 46th International ACM SIGIR Conference on Research and Development in Information Retrieval},
  pages={2639--2649},
  year={2023}
}

@inproceedings{geng2022recommendation,
  title={Recommendation as language processing (rlp): A unified pretrain, personalized prompt \& predict paradigm (p5)},
  author={Geng, Shijie and Liu, Shuchang and Fu, Zuohui and Ge, Yingqiang and Zhang, Yongfeng},
  booktitle={Proceedings of the 16th ACM conference on recommender systems},
  pages={299--315},
  year={2022}
}

@article{jang2016categorical,
  title={Categorical reparameterization with gumbel-softmax},
  author={Jang, Eric and Gu, Shixiang and Poole, Ben},
  journal={arXiv preprint arXiv:1611.01144
        
        
        
        
        
        
        
        
        
        },
  year={2016}
}

@article{rajput2023recommender,
  title={Recommender systems with generative retrieval},
  author={Rajput, Shashank and Mehta, Nikhil and Singh, Anima and Hulikal Keshavan, Raghunandan and Vu, Trung and Heldt, Lukasz and Hong, Lichan and Tay, Yi and Tran, Vinh and Samost, Jonah and others},
  journal={Advances in Neural Information Processing Systems},
  volume={36},
  pages={10299--10315},
  year={2023}
}

@article{liu2024generative,
  title={Generative Recommender with End-to-End Learnable Item Tokenization},
  author={Liu, Enze and Zheng, Bowen and Ling, Cheng and Hu, Lantao and Li, Han and Zhao, Wayne Xin},
  journal={arXiv preprint arXiv:2409.05546
        
        
        
        
        
        },
  year={2024}
}

@inproceedings{wang2024learnable,
  title={Learnable item tokenization for generative recommendation},
  author={Wang, Wenjie and Bao, Honghui and Lin, Xinyu and Zhang, Jizhi and Li, Yongqi and Feng, Fuli and Ng, See-Kiong and Chua, Tat-Seng},
  booktitle={Proceedings of the 33rd ACM International Conference on Information and Knowledge Management},
  pages={2400--2409},
  year={2024}
}

@inproceedings{lin2025order,
  title={Order-agnostic identifier for large language model-based generative recommendation},
  author={Lin, Xinyu and Shi, Haihan and Wang, Wenjie and Feng, Fuli and Wang, Qifan and Ng, See-Kiong and Chua, Tat-Seng},
  booktitle={Proceedings of the 48th international ACM SIGIR conference on research and development in information retrieval},
  pages={1923--1933},
  year={2025}
}
\end{document}